\documentclass[9pt,shortpaper,twoside,web]{ieeecolor}
\usepackage{generic}
\usepackage{cite}
\usepackage{cuted}
\usepackage{amsmath,amssymb,amsfonts}
\usepackage{algorithmic}
\usepackage{graphicx}
\usepackage{cite}
\usepackage{amsmath,amssymb,amsfonts}
\usepackage{graphicx,xcolor}
\usepackage{hyperref}
\usepackage{footmisc}
\usepackage{textcomp}
\usepackage{mathbbol}
\usepackage{dirtytalk}              
\usepackage{subcaption}
\usepackage{mathtools}
\usepackage{verbatim}
\usepackage{amssymb}
\let\labelindent\relax
\usepackage{enumitem}
\usepackage{arydshln}
\usepackage{epstopdf}

\newtheorem{definition}{Definition}
\newtheorem{remark}{Remark}
\newtheorem{thm}{Theorem}
\newtheorem{lemma}{Lemma}
\newtheorem{cor}{Corollary}

\newtheorem{expm}{Example}

\newtheorem{assump}{Assumption} 
\usepackage{textcomp}
\definecolor{hl}{RGB}{230,245,255} 
\def\BibTeX{{\rm B\kern-.05em{\sc i\kern-.025em b}\kern-.08em
    T\kern-.1667em\lower.7ex\hbox{E}\kern-.125emX}}
\begin{document}
\title{Topological Conditions for Echo Chamber Formation under the FJ model: A Cluster Consensus-based Approach %An Analysis for Cluster Consensus
%Echo chamber formation under the FJ model: A cluster consensus based approach
}
\author{Aashi Shrinate$^{1}$, \IEEEmembership{Student member, IEEE},  Twinkle Tripathy$^2$, \IEEEmembership{Senior Member, IEEE}, and Laxmidhar Behera$^{3}$ \IEEEmembership{Senior Member, IEEE}
\thanks{Aashi Shrinate$^{1}$ is a doctoral student and  Twinkle Tripathy$^{2}$ is an Assistant Professor in the Department of Electrical Engineering, Indian Institute of Technology Kanpur, Kanpur, Uttar Pradesh, India, 208016. \\
Laxmidhar Behera$^{3}$ is the director of the Indian Institute of Technology, Mandi. 
Email: {\tt\small aashis21@iitk.ac.in, ttripathy@iitk.ac.in and lbehera@iitk.ac.in}}}
\maketitle

\begin{abstract}
The Friedkin-Johnsen (FJ) model is a popular opinion dynamics model that explains the disagreement that can occur even among closely interacting individuals. Cluster consensus is a special type of disagreement, where agents in a network split into subgroups such that those within a subgroup agree and those in different subgroups disagree. {In large-scale social networks, users often distribute into echo chambers (\textit{i.e. groups of users with aligned views}) while discussing contested issues such as electoral politics, social norms, \textit{etc}}. Additionally, they are exposed only to opinions and news sources that align with their existing beliefs. Hence, the interaction network plays a key role in the formation of an echo chamber. Since cluster consensus can represent echo chambers in a social network, we examine the conditions for cluster consensus under the FJ model with the objective of determining the properties of the interaction network that lead to echo chamber formation. We present topology-based necessary and sufficient conditions for cluster consensus under the FJ model, regardless of the edge weights in the network and stubbornness values (which are difficult to estimate parameters in a social network). A major advantage of the proposed results is that they are applicable to arbitrary digraphs. Moreover, using the proposed conditions, we explain the emergence of bow-tie structures which are often observed in real-world echo chambers. Finally, we also develop a computationally feasible methodology to verify the proposed conditions for cluster consensus. 
\end{abstract}

 
\section{INTRODUCTION}
An echo chamber consists of users who are exposed only to partisan views that reinforce their existing beliefs \cite{pratelli2024entropy}. Empirical evidence shows that user interactions and recommendation algorithms play a significant role in the formation of the echo chambers \cite{echo_chamber}. This motivates the study of how social network interactions affect the emergence of echo chambers. 

The formation of echo chambers in social networks is equivalent to the networked agents achieving cluster consensus. The latter is a widely studied emergent behaviour in networked agents, that occurs when agents split into two or more subgroups such that the agents within a subgroup have consensus, but those in different ones disagree \cite{xia2011clustering}. Its weaker version is the multiconsensus problem \cite{gambuzza2020distributed}, where the agents in different subgroups can also sometimes have consensus. Cluster consensus finds several applications, such as task distribution in multi-robot systems \cite{Switching_T}, formation control \cite{anderson2008rigid}, and rendezvous at multiple points \cite{tomaselli2023multiconsensus}. In the sociological framework, cluster consensus represents the formation of subgroups in a society with aligned views (echo chambers), with widespread disagreement otherwise. %In general, networks of interacting agents (in biological swarms, power grids, drone swarms, \textit{etc.}) exhibit diverse emergent behaviours. Among these, cluster consensus is a widely studied outcome where distinct subgroups of agents converge to the distinct states, while achieving consensus within each subgroup. %Networks of interacting agents represent a wide range of real-world entities, including biological swarms, power grids, drone swarms, and social interactions. Local interactions among agents in such networks often result in diverse emergent behaviours. Among these, consensus is a widely studied outcome, where all the agents in a network converge to the same state. Cluster consensus is yet another interesting emergent behaviour,  %Therefore, cluster consensus captures a more realistic distribution of opinions in society than consensus, which rarely occurs. %In the literature, mathematical models, such as the DeGroot's model \cite{degroot1974consensus}, the FJ model \cite{friedkin1990opinions}, the {Hegselmann}-Krause model \cite{rainer2002opinion}, and the Biased assimilation model \cite{dandekar2013biased}, explain complex aspects of human interactions such as individual biases and homophily.

In the literature, several opinion dynamics models \cite{degroot1974consensus,friedkin1990opinions,rainer2002opinion,dandekar2013biased} characterise the evolution of opinions by accounting for complex aspects of human interactions. In the DeGroot model \cite{degroot1974consensus}, each agent updates its opinion as a weighted average of its neighbours' opinions. While this process of opinion evolution is supported by empirical evidence, it fails to explain the outcome of persistent disagreement in social networks. The FJ model \cite{friedkin1990opinions}, an extension of the DeGroot model, takes individuals' biases into account by introducing stubborn agents. The opinions of agents then generally converge to distinct values (reach disagreement). Other popular extensions of the DeGroot model include the Hegselmann-Krause (HK) \cite{rainer2002opinion} and the Biased Assimilation (BA) \cite{dandekar2013biased} model. The HK model incorporates the phenomenon of homophily, where each agent interacts only with others who share similar views. On the other hand, the BA model considers that an agent views information through a partisan lens, readily accepting information consistent with its beliefs. Among these models, the FJ model is popular because it has been empirically validated on large datasets \cite{friedkin2015control} and is also analytically tractable. Thus, in this paper, we explore the problem of cluster consensus under the FJ model. %Notably, the opinion dynamics models capture the collective decision-making in both socio-political and engineered networks. Thus, cluster consensus achieved under the FJ model can also be leveraged for engineering applications in robotic swarms, oscillator networks, \textit{etc.} %  in agent' interactions and is biased to explain disagreement. %In recent decades, social media networks have become a major platform for advertising, conducting socio-political campaigns, and spreading awareness.

\textit{Related Literature:} In the FJ framework, the authors in \cite{friedkin1990opinions} show that if each non-stubborn agent in the network is reachable from a stubborn one, then the final opinions of the entire network depend on the initial opinions of only the stubborn agents. Reference \cite{parsegov2016novel} introduces \textit{oblivious agents}, who are non-stubborn and do not have a path from any stubborn agent. The algebraic and graph-theoretic conditions for convergence in networks with oblivious agents are presented in \cite{parsegov2016novel} and \cite{TIAN2018213}, respectively. The works  \cite{friedkin2015control,como2016local,Opinion_Fluctuations,yao2022cluster} characterise the conditions that result in numerous emergent behaviours in the FJ framework. In \cite{friedkin2015control}, the authors design stubbornness values that lead to consensus and polarisation. However, stubbornness is an intrinsic property of an individual, which may not always be feasible to modify or even estimate. Reference \cite{como2016local} presents topology-based conditions under which polarisation occurs in undirected graphs with two stubborn agents. Further, in \cite{Opinion_Fluctuations}, the authors show that the opinions of non-stubborn agents can even exhibit oscillations under the FJ model with randomised interactions. In summary, the network topology plays a key role in shaping the emergent behaviours under the FJ framework \cite{como2016local,Opinion_Fluctuations}.
 
 %Further, the authors in \cite{Opinion_Fluctuations} examine the FJ model with randomised interactions and reveal that the opinions of non-stubborn agents can exhibit oscillations in the presence of stubborn agents.} In contrast, the present work focuses on {cluster consensus}, which remains less-explored in the FJ framework. %Cluster consensus occurs when the agents in a network split into two or more subgroups such that the agents within a subgroup have consensus, but those belonging to different groups disagree. Multiconsensus (or group consensus) is a special form of cluster consensus where the agents in different subgroups can sometimes have the same final opinion. Cluster consensus represents swarm behaviour and has applications in task distribution \cite{Switching_T} and formation control \cite{anderson2008rigid}. 
 
% In the opinion dynamics literature, cluster consensus has been observed under the Hegselmann-Krause model in \cite{rainer2002opinion}, where the opinion clusters form due to the homophilious interactions. 
 %The present work examines the role of the network topology in achieving cluster consensus in arbitrary digraphs. 
 More recently, the study of network topology in the FJ model is receiving attention in the design of recommender algorithms for opinion formation. The seminal works \cite{GHADERI20143209} and \cite{gionis2013opinion} quantify the influence of a stubborn agent using hitting probabilities of random walks over the network. These studies enable the design of topology-based interventions (recommendations) for desirably shaping opinions in \cite{gionis2013opinion,ECC_aashi}. Finally, reference \cite{yao2022cluster} presents topology-based conditions to achieve cluster consensus in the FJ framework. However, this result applies only to weakly connected digraphs, where each cycle contains atmost one stubborn agent.
 
 %In contrast, the present work focuses on {cluster consensus}, which remains less-explored in the FJ framework.
 \textit{Contributions:} In the present work, we consider a group of $n$ interacting individuals with their opinions governed by the FJ model. The primary objective of this work is to determine the topological properties of the interaction network that lead to the formation of echo chambers. To this end, we explore the relation between the network topology and final opinions in the FJ framework. %and identify the topological properties that lead to the equivalent phenomenon of cluster consensus.
 We start by presenting the necessary and sufficient conditions for any subgroup of agents in the network satisfies to always converge to the same final opinion (forms an opinion cluster). These conditions rely on the existence of special topological structures within the subgroups, defined by a \textit{Locally Topologically Persuasive (LTP) agent and its persuaded agent(s)} (introduced in our earlier work \cite{shrinate2025opinionclusteringfriedkinjohnsenmodel}). Interestingly, these LTP agents are not necessarily stubborn. Additionally, we prove that suitable edge-weights within the networks can also sometimes lead to the formation of opinion clusters. %\textcolor{red}{In contrast, if all the agents are non-stubborn, opinion clusters form even in the absence of the LTP agents.} 
 %We start by presenting the necessary and sufficient conditions for any subgroup of agents in the network satisfies to always converge to the same final opinion (forms an opinion cluster). TheseWhen a subgroup of agents containing a stubborn agent satisfies a special topological relation defined by a \textit{Locally Topologically Persuasive (LTP) agent and its persuaded agent(s)} (introduced in our earlier work \cite{shrinate2025opinionclusteringfriedkinjohnsenmodel}). Interestingly, LTP We show that suitable edge-weights within the networks can also sometimes lead to the formation of opinion clusters. %\textcolor{red}{In contrast, if all the agents are non-stubborn, opinion clusters form even in the absence of the LTP agents.} 
 
 %In the polarising discussions that cause echo chambers to form, . 
%In social media Echo chambers form due Often social media reinforces existing beliefs, causing like-minded users to form echo chambers.
%Social media platforms
%Due to tend interactions among like-minded user
%To maximise user engagement, social recommender systems often promote content that aligns with the user's beliefs. The lack of diverse perspectives causes groups of like-minded users to form echo chambers. Moreover, %, leading to political polarisation and extremist behaviours

Individuals in different echo chambers display sharp ideological divides, similar to the behaviour of cluster consensus in social networks. We establish that a network with $m$ stubborn agents achieves cluster consensus with $m$ opinion clusters if and only if it can be divided into $m$ subgroups, each composed only of a stubborn LTP agent and those persuaded by it. Contrary to \cite{yao2022cluster}, the proposed topology-based conditions can lead to cluster consensus in \textit{arbitrary digraphs}. We also present a computationally feasible methodology for verifying whether a network satisfies the proposed conditions. %Moreover, in the presence of oblivious agents, an additional opinion cluster emerges if the subgraph composed of these agents is aperiodic independly strongly connected compoenent.
%Therefore, next, we examine the emergence of cluster consensus (where two distinct opinion clusters never converge to the same opinion) under the FJ model. The presence of a stubborn agent in each opinion cluster aids the objective of attaining cluster consensus. We show that a network with $m$ stubborn agents achieves cluster consensus with $m$ opinion clusters if and only if it can be divided into $m$ subgroups, each composed only of a stubborn LTP agent and those persuaded by it. Moreover, in the presence of oblivious agents, an additional opinion cluster emerges if the subgraph of composed of these agents is aperiodic.
%We obtain an LTP-based necessary and sufficient condition on network topology that results in cluster consensus. 

 %the presence of stubborn agents within the subgroup plays a key role here. We show that , the stubborn agent and 

% Thereafter, we derive the conditions for cluster consensus  the agents partition into 

A preliminary version of this paper \cite{shrinate2025opinionclusteringfriedkinjohnsenmodel} (accepted for presentation at the American Control Conference, 2026), defined the notion of LTP agent and its topological relation with the agents in its vicinity that belong to its persuaded set. It established an LTP agent and those in its persuaded set, always form an opinion cluster.  The present work significantly generalises this sufficient condition by deriving necessary and sufficient conditions for the formation of opinion clusters, both in the presence and the absence of stubborn agents. Moreover, this work also presents necessary and sufficient conditions for cluster consensus to explain the formation of echo chambers.

The key contributions of this work can be summarised as follows:
\begin{itemize}
    %\item We define the notion of \textbf{\textit{ Locally Topologically Persuasive}} (LTP) agent and the unique set of non-influential agent(s) in its vicinity that the former is associated with. We present that an LTP agent and the agents associated with it always form an opinion cluster under the FJ model. Notably, this result highlights that even a set of non-influential agents with LTP-based topological properties can form an opinion cluster. %Furthermore, we use the LTP-based opinion clusters to present topology-based sufficient conditions for multiconsensus in any arbitrarily connected digraphs. 
    \item \textit{Formation of opinion clusters:} Under the FJ model, the LTP-based necessary and sufficient condition ensures that a subgroup of agents, containing a stubborn agent, forms an opinion cluster. We also show that a subgroup of agents composed only of non-stubborn agents can also form an opinion cluster if and only if a suitable edge weight relation holds. %S%pecial cases under which eqn. \eqref{eqn:rank_condition} holds are presented.
    \item \textit{Conditions for Cluster Consensus:} We present LTP-based necessary and sufficient conditions under which a network achieves cluster consensus under the FJ model.    This result generalises the conditions presented in \cite{yao2022cluster} to arbitrary digraphs and explains the bow-tie structure observed in real-world echo chambers \cite{mattei2022bow}.
    %We also show that if $\mathcal{G}$ contains the oblivious agents, it forms $m+1$ opinion clusters under cluster consensus.
   % su%\textit{if and only if}  the network $\mathcal{G}$ admits a partition into $m$ subgroups such that each group has a unique stubborn agent that isit has that form the necessary and sufficient conditions for a network with $m$ stubborn agents to achieve cluster consensus for \textit{almost all} initial states. 
    \item \textit{A methodology to verify proposed topological conditions}: We present a computationally efficient and scalable methodology that determines whether a network satisfies the LTP-based conditions for cluster consensus.%, with a time complexity of $O(n(n+|\mathcal{E}|))$, where $n$ and $|\mathcal{E}|$ denote the cardinality of nodes and edges, respectively. %in a graph.. %identify all the LTP agents in a network and determine its time complexity to show that identifying all LTP agents is computationally feasible (of order $O(n^3)$). Moreover, The time complexity of  is only of order 
    \end{itemize}

\textbf{Organisation of the Paper}: Sec. \ref{Sec2} introduces the notations and the relevant preliminaries. Sec. \ref{Sec:FJ} introduces the cluster consensus problem in the FJ framework. Sec. \ref{Sec:MR} presents LTP-based conditions for cluster consensus. Finally, Sec. \ref{Sec:conc} concludes with insights into future directions. The proofs of our results are given in the Appendix.

%In a preliminary version of this paper \cite{shrinate2025opinionclusteringfriedkinjohnsenmodel} (accepted for presentation at the American Control Conference, 2026), we introduced the Locally Topologically Persuasive (LTP) agents. An LTP agent is topologically suitably located such that, along with some agents in its vicinity (belonging in its persuaded set), it always converges to the same final opinion (leading to multiconsensus, a weaker form of cluster consensus). In this paper, we generalise these results by presenting necessary and sufficient conditions for cluster consensus, which  explain the formation of echo chambers in the FJ framework.
 \vspace{-7pt}
\section{Notations and Preliminaries}
\label{Sec2}
 
\textbf{Notations}:
%\subsection{Notations}
\label{subsec:NOT}
$\mathbb{R}$ denotes the set of real numbers. $\mathbb{1}_n$ ($\mathbb{0}_n$) is a column vector with each entry equal to $+1$ $(0)$ and $I_n$ is the identity matrix. $\mathbb{e}_i$ denotes the standard basis vector of $\mathbb{R}^n$ with $i^{th}$ entry equal to $1$.
Matrix $M=\operatorname{diag}(m_1,m_2,...,m_n)$ is a diagonal matrix.  $[n]$ denotes the set $\{1,2,...,n\}$ and $i:j$ denotes the set of natural numbers $\{i,i+1,...,j\}$. Consider a matrix $M \in \mathbb{R}^{n \times n}$, its spectral radius is represented by $\rho(M)$, its determinant is given as $\operatorname{det}(M)$ and its rank is $\operatorname{rank}(M)$. For index sets $\alpha,\beta \subseteq [n]$, the submatrix of $M$ with rows indexed by $\alpha$ and columns indexed by $\beta$ is given by $M[\alpha,\beta]$. If $\alpha=\beta$, we simply denote  $M[\alpha,\alpha]$ as $M[\alpha]$.
% 
%\subsection{Graph Preliminaries}

\textbf{Graph Preliminaries:}
Let $\mathcal{G}=(\mathcal{V},\mathcal{E})$ be a directed graph (digraph) with $\mathcal{V}=[n]$ denoting the $n$ agents and $\mathcal{E} \subseteq \mathcal{V} \times \mathcal{V}$ representing the set of edges. An edge from node $i$ to $j$ is denoted as $(i,j) \in \mathcal{E}$; it indicates that node $i$ is an \textit{in-neighbour} of node $j$ and node $j$ is an \textit{out-neighbour} of node $i$.
%The set of all in-neighbours of an agent $i$ is defined as $N_{in}(i)=\{j \in \mathcal{V}|(j,i) \in \mathcal{E}\}$ and the set of all of its out-neighbours is defined as $N_{o}(i)=\{j \in \mathcal{V}|(i,j) \in \mathcal{E}\}$.
%
A \textit{source} is a node without any in-neighbours.
The weighted adjacency matrix of $\mathcal{G}$ is denoted by $W=[w_{ij}] \in \mathbb{R}^{n \times n}$ with $w_{ij} > 0$ if $(j,i) \in \mathcal{E}$, zero otherwise. The in-degree of a node $p\in \mathcal{V}$ is defined as $d_{i}(p)=\sum_{j=1}^{n}w_{pj}$. The Laplacian matrix is defined as $L=D-W$, where $D=$ diag$(d_{i}(1),...,d_{i}(n))$.  {In \cite{Kron_red_digraphs}, the loopy Laplacian matrix $Q$ is defined as $Q :=L+$ diag$(w_{11},...,w_{nn})$. For any matrix $M=[m_{ij}]\in \mathbb{R}^{p \times p}$, we can construct its associated digraph $\mathcal{G}(M)$ consisting of $p$ nodes. Each non-zero entry $m_{ij}$ results in an edge $(j,i)$ in $\mathcal{G}(M)$. The matrix $Q$ is called the loopless Laplacian matrix as $(Q=L)$ if the associated digraph does not contain any self-loops.}

A \textit{walk} is an ordered sequence of nodes such that each pair of consecutive nodes forms an edge in the graph. A walk in which none of the nodes are repeated is referred to as \textit{path}. A \textit{cycle} is a walk whose initial and final nodes coincide. %{We say that a path from node $x$ to node $y$ traverses node $z$ if the sequence of nodes that form this path contains $z$}.
A graph is \textit{aperiodic} if the gcd of the lengths of all the simple cycles in the graph is $1$. An undirected graph is \textit{connected} if a path exists between any two nodes of the graph. A digraph where a directed path exists between any pair of nodes is called \textit{strongly connected}. {A digraph is \textit{weakly connected} if the undirected version of the digraph is connected.} The maximal subgraph of $\mathcal{G}$ which is strongly connected is called a \textit{strongly connected component} (SCC) of $\mathcal{G}$. An \textit{independent SCC} (iSCC) refers to a SCC in which the in-neighbour(s) of each node is within the SCC.

\textbf{Kron Reduction:}
%\textbf{Kron Reduction:}
\label{subsec:KR}
Let $M\in \mathbb{R}^{p \times p}$ with index set $\alpha\subseteq[p]$ and $\alpha^{c}:=[p] \setminus \alpha$. If $M[\alpha^c]$ is non-singular, then the Schur complement of $M[\alpha^c]$ in $M$ is defined as:
\begingroup
\setlength{\abovedisplayskip}{2pt}
\setlength{\belowdisplayskip}{2pt}
\begin{align}
\label{eqn:Schur_complement}
    M/\alpha^c=M[\alpha]-M[\alpha,\alpha^c](M[\alpha^c])^{-1}M[\alpha^c,\alpha]
\end{align}
\endgroup

 {Consider a set of linear equations $M\mathbf{y}=\mathbb{0}$, where $\mathbf{y}\in \mathbb{R}^{p}$. In \cite{zhang2006schur} (refer to Pg. 18), the author shows that eliminating $y[\alpha^c]$ using Schur complement \eqref{eqn:Schur_complement}, yields the reduced system of equations:
\begingroup
\setlength{\abovedisplayskip}{2pt}
\setlength{\belowdisplayskip}{2pt}
\begin{align}
\label{eqn:reduced_eqns_SC}
   M/\alpha^c \cdot \mathbf{y}[\alpha]=\mathbb{0}.
\end{align}
\endgroup
 
%The digraph $\mathcal{G}(Q)$ is associated with the loopy Laplacian matrix $Q$.
Under Kron Reduction, we reduce the digraph $\mathcal{G}(Q)$ associated with the loopy Laplacian matrix $Q$ by evaluating its Schur complement. The reduced graph $\mathcal{G}_{\alpha}$ obtained under Kron reduction is the associated digraph of $Q/\alpha^c$.

%, which has the following nice properties. %In \cite{Kron_red_digraphs}, the authors extended  {Kron Reduction} to digraphs for applications in the reduction of Markov chains.%, , with applications to the reduction of Markov chains. 
%\begin{comment}
  \begin{lemma}[\hspace{-0.01cm}\cite{Kron_red_digraphs}]
\label{lm:Basic_properties}
Consider a loopy Laplacian matrix $Q \in \mathbb{R}^{p}$ with $\alpha \subset [p]$. %Let $\mathcal{G}(Q)$ be the digraph associated with $Q$ and $\mathcal{G}_{\alpha}$ be the reduced digraph.
Under Kron reduction, the following properties hold:
\begin{enumerate}
    \item $Q/\alpha^c$ is well defined ({\textit{i.e.} $Q[\alpha^c]$ is invertible}) if for each node $i\in \alpha^c$ there is a node $j\in \alpha$ such that a path $j \to i$ exists in $\mathcal{G}(Q)$. (Note that, to adapt this condition to our framework, we consider $w_{ij}>0$ when $(j,i)\in \mathcal{E}$.)
    %  if $\alpha \subseteq [n]$.. %is associated \textbf{with final opinions of agents.}
    \item If $Q$ is a loopless Laplacian matrix, then $Q/\alpha^c$ is also a loopless Laplacian matrix.
    \item An edge $(j,i)$ exists in $\mathcal{G}_{\alpha}$ if and only if there is a path from $j$ to $i$ in $\mathcal{G}(Q)$ composed of nodes in set $\{i,j\} \cup \alpha^c$.
    %\item Consider $\alpha_1,\alpha_2\subset [p]$ such that $\alpha_2\subset \alpha_1$ and $\alpha_2^c=\alpha_1 \setminus \alpha_2$. %and$(Q/\alpha_1^c)/\alpha_2^c$, $Q/(\alpha_1^c \cup \alpha_2^c)$ are well-defined. 
    %Then, %$(Q/\alpha_1^c)/\alpha_2^c=Q/(\alpha_1^c \cup \alpha_2^c)$. 
    %\item The reduced graph $G_{\alpha}$ associated with the Laplacian matrix $Q/\alpha^c$ has an edge between two nodes $i$ and $j$ only if a path exists in $G(Q)$ from $i$ to $j$.
\end{enumerate}
\end{lemma}

%Consider an SCC composed of nodes in $\mathcal{G}$ such that all the in-neighbours of each node in the SCC also belong to the SCC. We refer to such an SCC as an independent strongly connected component (iSCC).  

%The condensation graph of a graph $\mathcal{G}$ is defined as $C(\mathcal{G})=(\mathcal{V}_c,\mathcal{E}_c)$. Each node $\mathcal{I} \in \mathcal{V}_c$ is a { \color{blue} strongly connected component (SCC)} of graph $\mathcal{G}$, and an edge $(\mathcal{I},\mathcal{J})\in \mathcal{E}_c $ exists if and only if an edge $(i,j)\in \mathcal{E}$ exists in graph $\mathcal{G}$ from node $i \in \mathcal{I}$ to a node $j \in \mathcal{J}$.  A sink of the condensation graph is {an \color{blue} SCC} that forms a node in the $C(\mathcal{G})$ without any outgoing edges.

\vspace{-5pt}
\section{Cluster Consensus in the FJ model}
\label{Sec:FJ}
In this section, we review the FJ model and highlight the relevance of examining cluster consensus under this model.
%Mathematically, we say that the individuals in a network of $n$ agents  have disagreement if there is a pair of individuals $i$ and $j$ such that the values $x_i \neq x_j$, where $x_i$ and $x_j$ denote their respective opinions.
%Lack of sufficient communication is not the only factor causing disagreement, as even well-connected individuals also commonly disagree.  
\begin{comment}
In social settings, discussions among individuals commonly result in disagreement over the well-explored consensus.
 The Friedkin-Johnsen (FJ) model is an averaging-based opinion dynamics model that captures the commonly occurring phenomenon of `persistent disagreement' in social networks. The distinctive property of this model is that it takes the agent's biases into account. Such agents who have a bias towards their preferences are referred to as the stubborn agents, and the rest are called non-stubborn.   
\end{comment}
\vspace{-7pt}
\subsection{FJ model}
\label{sec:FJ_model}
The FJ model is a popular opinion dynamics model based on iterative averaging that incorporates the heterogeneity among individuals due to their biases. Under this model, the individuals can exhibit varying degrees of resistance to changing their opinions. Any individual that does not readily change its opinion is called a stubborn agent with parameter $\beta_i\in [0,1]$ quantifying the degree of its stubbornness. 

Consider a network $\mathcal{G}$ of $n$ individuals (agents) with each individual's opinion about a topic denoted by a scalar $x_i$. Under the FJ model, the opinions of agents in $\mathcal{G}$ evolve as follows: 
\begingroup
\setlength{\abovedisplayskip}{2pt}
\setlength{\belowdisplayskip}{2pt}
\begin{align}
    \mathbf{x}(k+1)=(I_n-
    \beta)W \mathbf{x}(k)+\beta \mathbf{x}(0)
    \label{eq:FJ_opinion_dynamics}
\end{align}
\endgroup

where $\mathbf{x}=[x_1,...,x_n]$ and $\beta=$ diag$(\beta_1,...,\beta_n)$. An agent $i\in \mathcal{V}$ is stubborn if $\beta_i>0$.  {The matrix} $W$ is row-stochastic. 

Let $\mathcal{G}$ be a weakly connected network. In this scenario, $\mathcal{G}$ may contain agents that are unaffected by stubborn behaviour, \textit{i.e.}, a non-stubborn agent exists that does not have a path from any stubborn agent. Such agents are called \textit{oblivious agents} \cite{parsegov2016novel}. {We consider the nodes in $\mathcal{G}$ to be numbered such that the nodes $\{1,2,...,n_1\}$ are the non-oblivious agents and the remaining are oblivious, where $n_1\in[n]$. Further, let the iSCCs composed of the oblivious agents be referred to as oblivious iSCCs, and the agents forming such iSCCs have consecutive indices.} Now, we re-write the {FJ model} in \eqref{eq:FJ_opinion_dynamics} as:
\begingroup
\setlength{\abovedisplayskip}{2pt}
\setlength{\belowdisplayskip}{2pt}
\begin{align*}
\mathbf{x}_{1}(k+1)&=(I-\bar{\beta})\big(W_{11}\mathbf{x}_1(k)+ W_{12}\mathbf{x}_2(k)\big)+\bar{\beta}\mathbf{x}_{1}(0) \\
\mathbf{x}_2(k+1)&=W_{22}\mathbf{x}_2(k) 
\end{align*}
\endgroup
where $\mathbf{x}_1(k)\in \mathbb{R}^{n_1}$ and $\mathbf{x}_2(k)\in\mathbb{R}^{(n-n_1)}$ denote the opinions of the non-oblivious agents and oblivious agents, respectively, $W$ is partitioned as $W=\begin{bmatrix}
    W_{11} & W_{12} \\ \mathbb{0} & W_{22}
\end{bmatrix}$ and the diagonal matrix $\bar{\beta}$ denotes the stubbornness of non-oblivious agents. 
 {Throughout this article, we consider that $\mathcal{G}$ satisfies the following:
\begin{assump}
\label{assump:ns}
Each oblivious iSCC in $\mathcal{G}$, which is composed of two or more oblivious agents, is aperiodic.   
\end{assump}}

 {Assumption \ref{assump:ns} is the necessary and sufficient condition for convergence under the FJ model \cite{TIAN2018213}.}
%The following result gave the conditions under which the opinions of a weakly connected graph converge. % evolving und converge in a w
\begin{lemma}[\cite{parsegov2016novel,TIAN2018213}]
\label{lm:final_op}
 {Consider a network $\mathcal{G}$ with $m\geq 1$ stubborn agents such that Assumption \ref{assump:ns} holds.  The opinions evolving under the FJ model \eqref{eq:FJ_opinion_dynamics} converge to}
 \begingroup
\setlength{\abovedisplayskip}{2pt}
\setlength{\belowdisplayskip}{2pt}
 \begin{align}
 \label{eqn:final_op_with_oblivious}
 \mathbf{x}_1^*&=(I-(I-\bar{\beta})W_{11})^{-1} \big( (I-\bar{\beta}) W_{12}W_{22}^* \mathbf{x}_2(0) + \bar{\beta}\mathbf{x}_1(0)\big) \nonumber\\
  \mathbf{x}_2^*&= W_{22}^* \mathbf{x}_2(0) 
\end{align}
\endgroup
where $\mathbf{x}^*=[\mathbf{x}_1^*,\mathbf{x}_2^*]$ denotes the final opinions and  
$W_{22}^*=\lim_{k \to \infty}W_{22}^k$.
\end{lemma}

\begin{comment}
    In this article, we assume the following:

\begin{enumerate}
    \item[Assumption 1-] .
\end{enumerate}
The authors in \cite{} proved that Assumption-1 is the necessary and sufficient condition for the convergence of opinion under FJ model. 
\end{comment}
    
From Lemma \ref{lm:final_op}, it clearly follows that the final opinions depend on the initial opinions of only the stubborn agents and the  agents in oblivious iSCCs. 
Hence, these agents are referred to as the \textit{influential agents} and the set of all influential agents in $\mathcal{G}$ is denoted by $\mathcal{I}$.

Reference \cite{TIAN2018213} establishes that the final opinions of agents lie in the convex hull of the initial opinions of influential agents.   Moreover, the coefficients of each of these convex combinations (that determine the final opinions) are related to the network topology (captured by $W$) and the stubbornness of agents \cite{GHADERI20143209}. Depending on these factors, the final opinions of agents (except the agents in an oblivious iSCC) often converge to arbitrary distinct values, which is termed as disagreement. 
\vspace{-25pt}
\subsection{Cluster Consensus}
{ In social networks, disagreement is more prevalent as opposed to consensus \cite{friedkin2015control}. 
A special form of disagreement known as `echo chambers' is commonly observed across large-scale social networks \cite{echo_chamber}, where the agents divide into subgroups such that the agents within each subgroup have the same opinions. In the literature on multi-agent systems, this behaviour is termed as multiconsensus (or 
group consensus).} More formally, multiconsensus is defined as: 
%An opinion cluster is defined as:}
\begin{definition} Consider a partition of the agents in $\mathcal{G}$ into $z$ disjoint subgroups (clusters) $\mathcal{V}_1,...,\mathcal{V}_z$ such that $\cup_{h=1}^z \mathcal{V}_h=\mathcal{V}$. The network achieves \textit{multiconsensus} if, for any initial state $\mathbf{x}(0)\in \mathbb{R}^{n}$, the final opinions of agents in each $\mathcal{V}_h$ satisfy:
%A set of agents $\mathcal{C}\subset \mathcal{V}$ form an opinion cluster if 
\begingroup
\setlength{\abovedisplayskip}{2pt}
\setlength{\belowdisplayskip}{2pt}
\begin{align}
\label{eqn:opinion_cluster}
    \lim_{k \to \infty} x_i(k)-x_j(k)&=0 \qquad \forall \ i,j\in \mathcal{V}_h 
\end{align}
\endgroup
for each $h\in[z]$. A set of agents that satisfy eqn. \eqref{eqn:opinion_cluster} for all initial states is called an \textit{opinion cluster}. Clearly, each subgroup $\mathcal{V}_h$ forms an {opinion cluster}.
%where $x_k^*$ denotes the final opinion of agent $k\in \mathcal{V}$.    
\end{definition}

Under multiconsensus, two or more different subgroups can converge to the same final opinion. When the final opinions of agents in different subgroups are always distinct, \textit{cluster consensus} occurs. Therefore, under cluster consensus, in addition to eqn. \eqref{eqn:opinion_cluster}, the following holds:
\begingroup
\setlength{\abovedisplayskip}{2pt}
\setlength{\belowdisplayskip}{2pt}
\begin{align}
\label{eqn:distinct_op}
  \lim_{k \to \infty} x_i(k) -x_j(k)&\neq 0 \qquad \text{if } \  i\in \mathcal{V}_h ,j\notin \mathcal{V}_h,  \quad \forall h\in\ [z]
\end{align}
\endgroup
%In this framework, a subgroup of agents converging to the same opinion is called an opinion cluster.

%Thus, . %not only the agents within each subgroup have the same opinion, but agents in two different subgroups must have different final opinion. U
%

%Opinion clustering refers to the formation of one or more opinion clusters in the network. Two special cases of opinion clustering are defined as follows:

%It can be of the following types:

% Therefore,  we examine the topology-based conditions that result in cluster consensus.
 
 %While homophily is a major cause of echo chamber formation, . While 

%Under the FJ model, the final opinion of each agent is governed by the initial opinions of the influential agent and lies in the convex hull formed by them.Consider a pair of agents $i$ and $j$ such that each influential agent in the network has an equal impact on both. Then, their final opinions satisfy eqn. \eqref{eqn:opinion_cluster} and converge to the same opinion, regardless of the initial opinions. From eqn. \eqref{eqn:final_op_with_oblivious}, the impact of the influential agents on others depends on network topology, edge weights and stubbornness. Thus, byfor a suitably choosing these parameters a network can be designed where a set of agents form an opinion cluster. However, identifying such parameters is generally a challenging problem.
%However, identifying the conditions is non-trivial.
 %, thereby identify the conditions under which each influential agent has an equal impact

In \cite{yao2022cluster}, the authors derive topology-based conditions for cluster consensus under the FJ model, while assuming that each cycle in the digraph contains atmost one stubborn agent. Consequently, these conditions do not explain the emergence of cluster consensus in general digraphs. In this work, our key objective is to address this limitation by deriving the generalised conditions under which the agents achieve cluster consensus under the FJ model. %without considering this assumption.
 \vspace{-5pt}
\section{Conditions for Cluster Consensus}
\label{Sec:MR}
%In this section, we begin by defining the LTP agents, first introduced in our work \cite{}. Thereafter, leveraging the notion of LTP agents,  we present a partitioning of the network that achieves cluster consensus.
 
%introduce the notion of LTP agents and describe their salient properties in the next section.we employ the  {Kron Reduction} technique to derive topology-based %sufficient and necessary conditions under which a set of agents in a network form an opinion cluster %have cluster consensus %Consider a set of agents in $\mathcal{G}$ such that under the FJ model. Thereafter, 

In this section, we begin by exploring the conditions that lead to a set of agents forming an opinion cluster under the FJ model. Thereafter, we present the necessary and sufficient conditions for a network to achieve cluster consensus. %begin by stating the conditions under which 
\vspace{-10pt}
\subsection{Locally Topologically persuasive  (LTP) agent}
\label{subsec:LTP}
%In the FJ framework, the final opinion of each agent is governed by the initial opinions of the influential agents and lies in the convex hull formed by them.
%\textcolor{red}{Consider a set of agents $\mathcal{C} \subset \mathcal{V}$. Suppose each influential agent in the network has an equal impact on the final opinion of each agent in a set $\mathcal{C}$. Then, the final opinions of all agents in $\mathcal{C}$ are equal, regardless of the initial opinions, and $\mathcal{C}$ forms an opinion cluster.} % In , the authors show that
{As discussed in Sec. \ref{sec:FJ_model}, under the FJ model, the final opinion of each agent is shaped by the influential agents. In \cite{GHADERI20143209}, the authors show that the impact of a stubborn agent on the final opinion of another agent depends on the weights of edges on each path from the former to the latter. Hence, the final opinion of an agent is determined by the network interconnections and the associated edge weights.}
%The formation of an opinion cluster by

{A subset of agents $\mathcal{C} \subset \mathcal{V}$ forms an opinion cluster if  their final opinions are equal, for any choice of initial opinions. Here, the exact value to which the final opinion of all agents in $\mathcal{C}$ converge is not of concern. Can simple topology-based conditions ensure that agents in $\mathcal{C}$ form opinion cluster, regardless of the edge weights in the network?} %does not depend on the exact value However, for the formation of an opinion cluster by the agents in $\mathcal{C}$, their final opinions these agents must be equal, but exact value is not of concern.
%exact value of the final opinion depends on the imNote that The formation of  $\mathcal{C}$ requires that the impact of each influential agent must be equal on each agent in $\mathcal{C}$ but does not depend on its exact value. 
%depends on both the paths and respective edge weights, formation of an opinion cluster requires this impact to be equal for each agent in $\mathcal{C}$, independent of the exact value of the impact.
%Can this impact be equal just based on the network topology, regardless of the edge weights.\textcolor{magenta}{However, the topological conditions on the paths that ensure that the impact of each influential agent on distinct agents is equal, leading to the formation of an opinion cluster, remains unknown}. 
In \cite{shrinate2025opinionclusteringfriedkinjohnsenmodel}, we answer this question positively by introducing the notion of an LTP agent and those it persuades. Moreover, we established that the `nice' topological properties of these agents result in the formation of an opinion cluster.

\begin{definition} \cite{shrinate2025opinionclusteringfriedkinjohnsenmodel}
\label{defn:LTP}%Consider an agent $p$ and a non-influential agent $q$ 
%Consider an agent $p$ in network $\mathcal{G}$. 
An agent $p\in \mathcal{V}$ is called an {LTP agent} if the following hold:
\begin{enumerate}
    \item[(i)] there is an agent $q\in \mathcal{V}\setminus \mathcal{I}$ such that every path in $\mathcal{G}$ {(that exists)} from an influential agent $s\in \mathcal{I}$ to $q$ (where $p \neq q$) traverses $p$,
    \item[(ii)]  {additionally, if $p \in  \mathcal{V}\setminus \mathcal{I}$, an agent $c\in \mathcal{V}$ such that all paths from each agent $s \in \mathcal{I}$ to $p$ traverse $c$, does not exist.}
\end{enumerate}
If $p$ is an LTP agent, then the agent $q$ satisfying condition (i) is said to be {\textit{persuaded by}} $p$. The set of all non-influential agents persuaded by an LTP agent $p$ is denoted by $\mathcal{N}_p$.
\end{definition}

  Note that if $p\in \mathcal{I}$, condition (i) is sufficient to ensure that $p$ is an LTP agent. Condition (ii) in Defn. \ref{defn:LTP} is applicable only if $p$ is a non-influential agent.  It conveys that a non-influential agent $p$, persuaded by an LTP agent $c$ (\textit{i.e} $p\in \mathcal{N}_{c}$), is not an LTP agent itself.} 
The following example illustrates an LTP agent in a network and those it persuades.
\begin{figure}[h]
    \centering
 \begin{subfigure}{0.17\textwidth}
        \centering
    \includegraphics[width=0.75\linewidth]{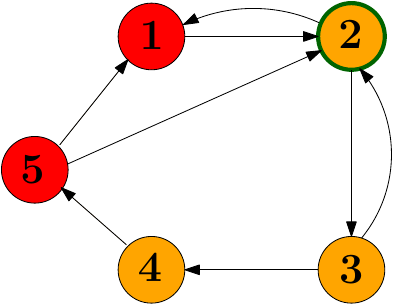}
    \caption{Network $\mathcal{G}$ has stubborn agents $1$ and $5$ }
    \label{Fig:network_g_LTP}
    \end{subfigure}
    \begin{subfigure}{0.13\textwidth}
        \centering
    \includegraphics[width=1\linewidth]{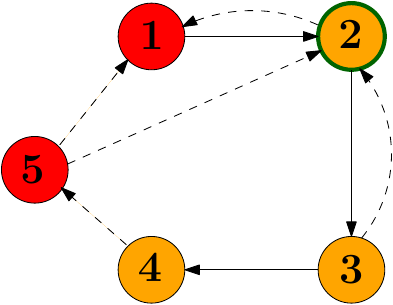}
    \caption{Paths from $1$ to $3$ (and $4$) traverse $2$}
       \label{Fig:LTP_expm_1}
    \end{subfigure}
    \hfill
     \begin{subfigure}{0.13\textwidth}
        \centering
    \includegraphics[width=1\linewidth]{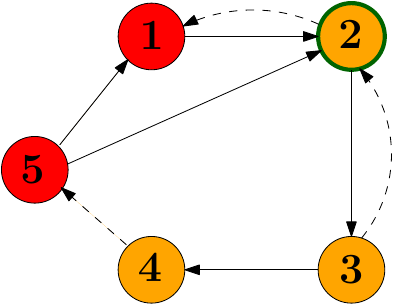}
    \caption{Paths from $5$ to $3$ (and $4$) traverse $2$}
       \label{Fig:LTP_expm_2}
    \end{subfigure}
    \caption{In the given figure, agent $2$ is an LTP agent with $\mathcal{N}_2=\{3,4\}$.}
    \label{fig:LTP_example}
\end{figure}
\begin{expm}
\label{expm:2}
 Consider the network $\mathcal{G}$ shown in Fig. \ref{Fig:network_g_LTP}. $\mathcal{G}$ is a strongly connected network, hence, only the stubborn agents $1$ and $5$ are influential. Fig \ref{Fig:LTP_expm_1} demonstrates that each path in $\mathcal{G}$ from $1$ to $3$ (and $4$) traverses $2$ (highlighted by solid lines). Fig. \ref{Fig:LTP_expm_2} shows that the paths from $5$ to $3$ (and $4$) also traverse $2$. Thus, for agent $2$, condition i) in Defn. \ref{defn:LTP} holds. Further,
 since each path from $1$ and $5$ to $2$ do not traverse any common node, hence, condition ii) in Defn. \ref{defn:LTP} also holds. Thus, $2$ is an LTP agent and $\mathcal{N}_2=\{3,4\}$. %Similarly, $6$ is also an LTP agent and $\mathcal{N}_6=\{1\}$.
Note that each path from the stubborn agents $1$ and $5$ to $4$ traverses $3$ as well, however $3$ is not an LTP agent because $3\in \mathcal{N}_2$.
\end{expm}

%Since the existence of an LTP agent depends only on the paths from the influential agent to a non-influential agent, it is a graph topological property.
By definition, an LTP agent and those persuaded by it have the following additional properties:
% Next, we present some salient properties of the LTP agents:
 \begin{enumerate}[leftmargin=0pt, itemindent=1.5em]
     \item \textit{In-neighbours}: Let $p$ be an LTP agent $p$ and $q \in \mathcal{N}_p$, then each in-neighbour of $q$ lies in the set $\mathcal{N}_p \cup \{p\}$; otherwise, a path exists from an influential agent that does not traverse $p$. On the contrary, $p$ can have any agent in $\mathcal{G}$ as its in-neighbour.
     \item \textit{Unique persuaded agents:} If a network has multiple LTP agents, the sets of agents persuaded by any two distinct LTP agents are disjoint \cite{shrinate2025opinionclusteringfriedkinjohnsenmodel}. Therefore, excluding agents that are neither an LTP agent nor are persuaded by any LTP agent, we can uniquely partition the rest into disjoint groups, with each group consisting of one LTP agent and its persuaded agents.
     \item \textit{LTP agents vs cut vertices:} The notion of an LTP agent is similar to that of a cut vertex in undirected graphs (and strong articulation points in digraphs). A node $v$ is a cut vertex (or a strong articulation point in a strongly connected digraph) if and only if there are distinct nodes $x$ and $y$ such that every path from $x$ to $y$ traverses $v$ \cite{ITALIANO201274}. Since LTP agents are defined for digraphs, each LTP agent in a strongly connected digraph is a strong articulation point; however, the converse does not hold. (e.g., node~$3$ in Example~\ref{expm:2}).

 \end{enumerate}

\begin{lemma}(\cite{shrinate2025opinionclusteringfriedkinjohnsenmodel})
\label{lemma:ACC}
   Consider a network $\mathcal{G}=(\mathcal{V},\mathcal{E})$ with $m\geq 1$ stubborn agents such that the Assumption \ref{assump:ns} holds. Let the opinions of agents in $\mathcal{G}$ be governed by the FJ model \eqref{eq:FJ_opinion_dynamics}. If $p\in \mathcal{V}$ is an LTP agent, then $p$ and  the nodes in set $\mathcal{N}_p$ form an opinion cluster. 
\end{lemma}

In the next set of results, we generalise the sufficient condition presented in Lemma \ref{lemma:ACC} for the formation of an opinion cluster. %following results, we further generalise this result.
\vspace{-15pt}
\subsection{Formation of opinion clusters}
Consider a network $\mathcal{G}$ with $m$ stubborn agents and $u\geq 0$ oblivious iSCCs. It is already known that the agents in each oblivious iSCC evolve under the DeGroot model and form an opinion cluster \cite{bullo}. \textit{How do the remaining agents, who are not a part of any oblivious iSCCs, behave? Since the FJ model predominantly explains disagreement, that is the most expected outcome. Interestingly, we show that the system dynamics coupled with the network topology can lead to opinion clustering within the larger group, a special kind of disagreement.} In this subsection, we aim to ascertain the conditions under which a subset of agents $\mathcal{C}\subseteq \mathcal{V}$ form an opinion cluster. Without loss of generality, we index the agents such that:
(i) the agents in $\mathcal{C}$ have indices $\{1,\dots,|\mathcal{C}|\}$,
(ii) they are followed by agents (if any) that are neither in $\mathcal{C}$ nor in any oblivious iSCCs, and
(iii) the remaining agents belong to the $u$ oblivious iSCCs, with agents in the same iSCC assigned consecutive indices.
With this ordering, the opinion vector is partitioned as
\begingroup
\setlength{\abovedisplayskip}{2pt}
\setlength{\belowdisplayskip}{2pt}
\[
{\mathbf{x}}(k)=
\begin{bmatrix}
\hat{\mathbf{x}}_1(k) &
\hat{\mathbf{x}}_2(k)&
\cdots&
\hat{\mathbf{x}}_{u+2}(k)
\end{bmatrix}^T
\]
\endgroup

where $\hat{\mathbf{x}}_1(k)$  {denotes the opinions} of agents in $\mathcal{C}$,
$\hat{\mathbf{x}}_2(k)$  {denotes the opinions} of agents neither in $\mathcal{C}$ nor in any oblivious iSCC,
and $\hat{\mathbf{x}}_i(k)$  {denotes the opinions} of the oblivious agents in the $(i-2)^{th}$ iSCC for $i\in \{3,...,u+2\}$. 
%We renumber the agents in $\mathcal{G}$ such that the agents in $\mathcal{C}$ are indexed $\{1,...,|\mathcal{C}|\}$, followed by those agents not any iSCCs of oblivious agents and which are followed by oblivious agents in iSCCs such that those belonging to the same iSCC have consecutive indices. 
Then, we can decompose eqn. \eqref{eq:FJ_opinion_dynamics} as:
\begingroup
\setlength{\abovedisplayskip}{2pt}
\setlength{\belowdisplayskip}{2pt}
\begin{align*}
 \hat{\mathbf{x}}_1(k+1)&= (I-\beta_{11})( \sum_{j=1}^{u+2}\hat{W}_{1j} \hat{\mathbf{x}}_j(k)   )+\beta_{11} \hat{\mathbf{x}}_1(0)  \\
  \hat{\mathbf{x}}_2(k+1)&= (I-\beta_{22})(\sum_{j=1}^{u+2}\hat{W}_{2j} \hat{\mathbf{x}}_j(k))+\beta_{22} \hat{\mathbf{x}}_2(0)
\end{align*}
\endgroup

\begingroup
\setlength{\abovedisplayskip}{2pt}
\setlength{\belowdisplayskip}{2pt}
\begin{align}
  \hat{\mathbf{x}}_3(k+1)&=\hat{W}_{33} \hat{\mathbf{x}}_3(k) \nonumber \\
 &\vdots \nonumber \\
\hat{\mathbf{x}}_{u+2}(k+1)&=\hat{W}_{(u+2)(u+2)} \hat{\mathbf{x}}_{u+2}(k) \label{eqn:2_cluster_ss}
\end{align}
\endgroup
%where $\mathbf{x}_1$ denote the opinions of agents in $\mathcal{C}$, $\mathbf{x}_2$ denote the opinions of agents neither in any oblivious iSCCs nor in $\mathcal{C}$ and $\mathbf{x}_i$ denotes opinions of oblivious agents in $i^{th}$ iSCC. Similarly,  Equivalently,
where $\beta_{11}$ and $\beta_{22}$ are diagonal matrices that denote the stubbornness of agents in $\mathcal{C}$ and those neither in $\mathcal{C}$ nor any oblivious iSCCs, respectively. Based on the new indexing, $W$ is partitioned as:
\begingroup
\setlength{\abovedisplayskip}{2pt}
\setlength{\belowdisplayskip}{2pt}
\begin{align*}
W=\begin{bmatrix}
    \hat{W}_{11} & \hat{W}_{12} & \hat{W}_{13} & \cdots & \hat{W}_{1(u+2)} \\
     \hat{W}_{21} & \hat{W}_{22} & \hat{W}_{23} & \cdots & \hat{W}_{2(u+2)} \\
     \mathbb{0} & \mathbb{0} &  \hat{W}_{33} & \mathbb{0} & \mathbb{0}\\
     \vdots & \vdots & \cdots & \cdots & \vdots\\
      \mathbb{0} & \mathbb{0} &  \mathbb{0} & \mathbb{0} & \hat{W}_{(u+2)(u+2)}\\
\end{bmatrix}    
\end{align*}
\endgroup
%As discussed in \ref{subsec:LTP},\textcolor{magenta}{where?}, the network topology plays a key role in the formation of final opinions under the FJ model. 
%However, 
Lemma \ref{lemma:ACC} showed that if $\mathcal{C}$ is composed of an LTP agent and its persuaded agents, it forms an opinion cluster for any value of the edge weights and stubbornness. The following result explores the following: \textit{Is the LTP-based condition both necessary and sufficient to form an opinion cluster?} \textit{If not, can we obtain even more general conditions for forming opinion clusters?}
%The following result The following result presents conditions based on topology and edge weights (of certain incoming edges to agents in $\mathcal{C}$) in the network under which the agents in $\mathcal{C}$ form an opinion cluster. \textcolor{magenta}{This is still a weak motivation, rephrase it.}

\begin{thm}
\label{thm:WC}
    Consider a network $\mathcal{G}=(\mathcal{V},\mathcal{E})$ with $m\geq 1$ stubborn agents labelled as $s_1,...,s_m$ that satisfies Assumption \ref{assump:ns}. Let the opinions of agents in $\mathcal{G}$ be governed by the FJ model \eqref{eq:FJ_opinion_dynamics}. Then, a set of agents $\mathcal{C}\subseteq\mathcal{V}$ forms an opinion cluster (equivalently, satisfies eqn. \eqref{eqn:opinion_cluster}) if and only if, any one of the following conditions hold:  
 \begin{enumerate}
     \item[C1.]  $\mathcal{C}$ does not contain any stubborn agents and  \begingroup
\setlength{\abovedisplayskip}{2pt}
\setlength{\belowdisplayskip}{2pt}
\begin{equation}
     \label{eqn:rank_condition}
\begin{aligned}
\operatorname{rank} \Big( \big[ 
    &\hat{W}_{13}\mathbb{1} + Z(I-\beta_{22})\hat{W}_{23}\mathbb{1} \quad \dots \\
    &\hat{W}_{1(u+2)}\mathbb{1} + Z(I-\beta_{22})\hat{W}_{2(u+2)}\mathbb{1} \quad  Z\psi \big] \Big) = 1
\end{aligned}
\end{equation}
\endgroup
   where $Z=W_{12}\big(I-(I-\beta_{22})\hat{W}_{22}\big)^{-1}$ and $\psi=[\psi]_{ij}$ is a matrix with entry $\psi_{ij}=\beta_i$ if  stubborn agent $i\notin\mathcal{C}$ has label $s_j$, otherwise $\psi_{ij}=0$.
     
     \begin{comment}
         
     $\operatorname{rank}\bigg(\begin{bmatrix}
            \hat{W}_{12} & \hat{W}_{13} \mathbb{1} & \cdots & \hat{W}_{1(u+2)} \mathbb{1}
         \end{bmatrix}\bigg)=1$ 
     \end{comment}

     \item[C2.] $\mathcal{C}$ contains a unique stubborn agent $p\in \mathcal{V}$ which is an LTP agent and $\mathcal{C}=\{p\} \cup \mathcal{N}_p$.
 \end{enumerate}   
\end{thm}
The proof of Theorem \ref{thm:WC} is given in the Appendix.

\begin{remark}
{Theorem \ref{thm:WC} shows that any set of agents $\mathcal{C}$ that forms an opinion cluster can have atmost one stubborn agent. C2.) is the necessary and sufficient condition that ensures that a set $\mathcal{C}$ containing a stubborn agent forms an opinion cluster. This condition is purely based on graph topology as it depends only on the paths from the influential agents outside $\mathcal{C}$ to non-influential agents within $\mathcal{C}$. Hence, the agents satisfying C2.) form an opinion cluster independent of any edge weights in the network and stubbornness values. }

On the contrary, C1.) depends on the edge weights of all interactions except those within the cluster (denoted by $\hat{W}_{11}$) and on the stubbornness $\beta_{22}$. Interestingly, eqn. \eqref{eqn:rank_condition} holds under the following rank condition on the edge weights of the incoming edges to $\mathcal{C}$ from agents not in $\mathcal{C}$, \begingroup
\setlength{\abovedisplayskip}{2pt}
\setlength{\belowdisplayskip}{2pt}
\begin{align}
\label{eqn:rank_condition_sc}
    \operatorname{rank}\bigg(\begin{bmatrix}
            \hat{W}_{12} & \hat{W}_{13} \mathbb{1} & \cdots & \hat{W}_{1(u+2)} \mathbb{1}
         \end{bmatrix}\bigg)=1
\end{align}
\endgroup
Thus, eqn. \eqref{eqn:rank_condition_sc} becomes a special case of C1.) that depends only on a limited number of edge weights in the network. %Under this condition, the agents in $\mathcal{C}$ form an opinion cluster if the weights of incoming edges to agents in $\mathcal{C}$ from any two agents not in $\mathcal{C}$ nor in oblivious iSCCs have the same ratio. Similarly, 

%such that the weights of incoming edges from agents is in the same ratio, form an opinion cluster that have incoming  from the rest

Additionally, as shown in \cite{shrinate2025opinionclusteringfriedkinjohnsenmodel}, a set $\mathcal{C}=\{p\} \cup \mathcal{N}_p$ composed only of non-stubborn agents, where $p$ is a non-stubborn LTP agent, also forms an opinion cluster.
This condition is also a special case of C1.) because only $p$ can have incoming edges from agents not in $\mathcal{C}$. WLOG, we can index $p$ as $1$, hence, $\hat{W}_{1j}$ is of the form $\hat{W}_{1j}=\mathbf{b}\mathbf{a}_j^T$ with $\mathbf{b}\in \mathbb{R}^{|\mathcal{C}|}=[1,0,...,0]^T$ and $\mathbf{a}_j$ is a real vector of appropriate dimensions, $j\in\{2,...,u+2\}$. Hence, the rank condition \eqref{eqn:rank_condition_sc} holds.

%In the upcoming dicussion, we show that the LTP-based condition is itself a special case of eqn. \eqref{eqn:rank_condition_sc}.

\end{remark}

The following example illustrates the conditions in Theorem \ref{thm:WC}.
\begin{expm}
Consider the network $\mathcal{G}$ shown in Fig. \ref{fig:Network_g} with stubborn agents $1$ and $4$. The agents $7,8$ and $9$ form an oblivious iSCC in $\mathcal{G}$. Let agents from $\mathcal{G}$ be selected such that $\mathcal{C}=\{2,6\}$. For this choice of $\mathcal{C}$, we renumber the nodes in $\mathcal{G}$ such that: (1) first, the agents in $\mathcal{C}$ are indexed $\hat{1}=2,\hat{2}=6$ (ii) followed by agents who are neither in $\mathcal{C}$ nor in oblivious iSCCs, $\hat{3}=1,\hat{4}=3,\hat{5}=4,\hat{6}=5$ and (iii) followed by the agents in the oblivious iSCC $\hat{7}=7,\hat{8}=8,\hat{9}=9$.   
Based on this renumbering, we partition $W$ and obtain
\begingroup
\setlength{\abovedisplayskip}{2pt}
\setlength{\belowdisplayskip}{2pt}
\begin{align*}
\hat{W}_{12}=\begin{bmatrix}
    0.2 & 0 & 0 & 0.3\\
    0.1 & 0 & 0 & 0.15
\end{bmatrix}, \qquad \hat{W}_{13}\mathbb{1}=\begin{bmatrix}
    0.4 \\0.2
\end{bmatrix}    
\end{align*}
\endgroup
Thus, $\begin{bmatrix}
   \hat{W}_{12} & W_{13}\mathbb{1}
\end{bmatrix}$ forms a rank-1 matrix. Similarly, for $\mathcal{C}=\{3,5\}$, we again re-number the agents and obtain 
\begingroup
\setlength{\abovedisplayskip}{2pt}
\setlength{\belowdisplayskip}{2pt}
\begin{align*}
 \hat{W}_{12}=\begin{bmatrix}
    0 & 0.5 & 0.5 & 0\\
    0 & 0.5 & 0.5 & 0
\end{bmatrix}, \qquad  \hat{W}_{13}\mathbb{1}=\begin{bmatrix}
    0 \\
    0
\end{bmatrix}  
\end{align*}
\endgroup
Here as well, $\operatorname{rank}(\begin{bmatrix}
   \hat{W}_{12} & \hat{W}_{13}\mathbb{1}
\end{bmatrix})=1$. By Theorem \ref{thm:WC}, each set $\{2,6\}$ and $\{3,5\}$ forms an  opinion cluster, as shown in Fig. \ref{fig:plot_opinion_clusters}.

%$\mathcal{G}$ can renumbered such that: the nodes in $\mathcal{C}$such that $\hat{W}_{12}=[]$, $\hat{W}_{13}\mathbb{1}=$. 

\begin{comment}
   the sub-matrices $W[\{3,8\},\{1,2,4,5,6,7\}]$ and $W[\{4,7\},\{1,2,3,5,6,8\}]$ have rank equal to $1$.
\begin{align*}
 &W[\{3,8\},\{1,2,4,5,6,7\}]=\begin{bmatrix}
    0.15 & 0.1 & 0 & 0 & 0& 0.25\\
    0.3 & 0.2 & 0 & 0 & 0 & 0.5
\end{bmatrix} \\
&W[\{4,7\},\{1,2,3,5,6,8\}]=\begin{bmatrix}
    0 & 0 & 0.5 & 0 &0.5 & 0 \\
    0 & 0 & 0.5 & 0 & 0.5 & 0 \\
\end{bmatrix}
\end{align*}
Hence, for node sets $\{3,8\}$ and $\{4,6\}$, condition C1.) holds. Therefore, under the FJ model, the agents $3,8$ form one opinion cluster and the agenta $4,7$ form another opinion cluster, as shown in Fig. \ref{fig:plot_opinion_clusters}. Further, node $5$ is an LTP agent with $\mathcal{N}_5=\{6\}$, hence, $5,6$ form the third opinion cluster.      
\end{comment}

\end{expm}

\begin{figure}[h]
    \centering
    \begin{subfigure}{0.23\textwidth}
        \centering
        \includegraphics[width=01\linewidth]{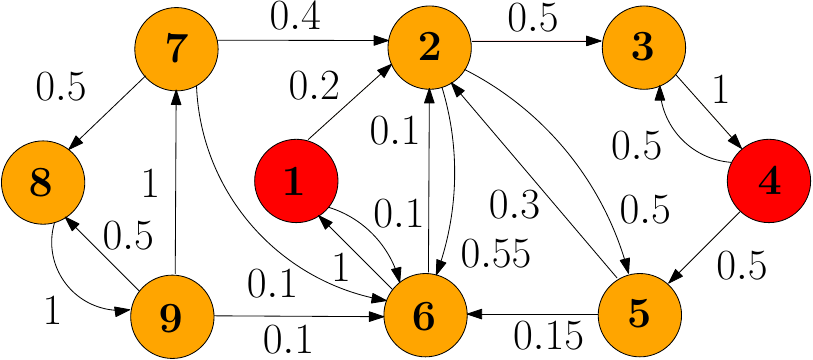}
    \caption{Network $\mathcal{G}$}
    \label{fig:Network_g}
    \end{subfigure}
    \begin{subfigure}{0.23\textwidth}
        \centering
        \includegraphics[width=1\linewidth]{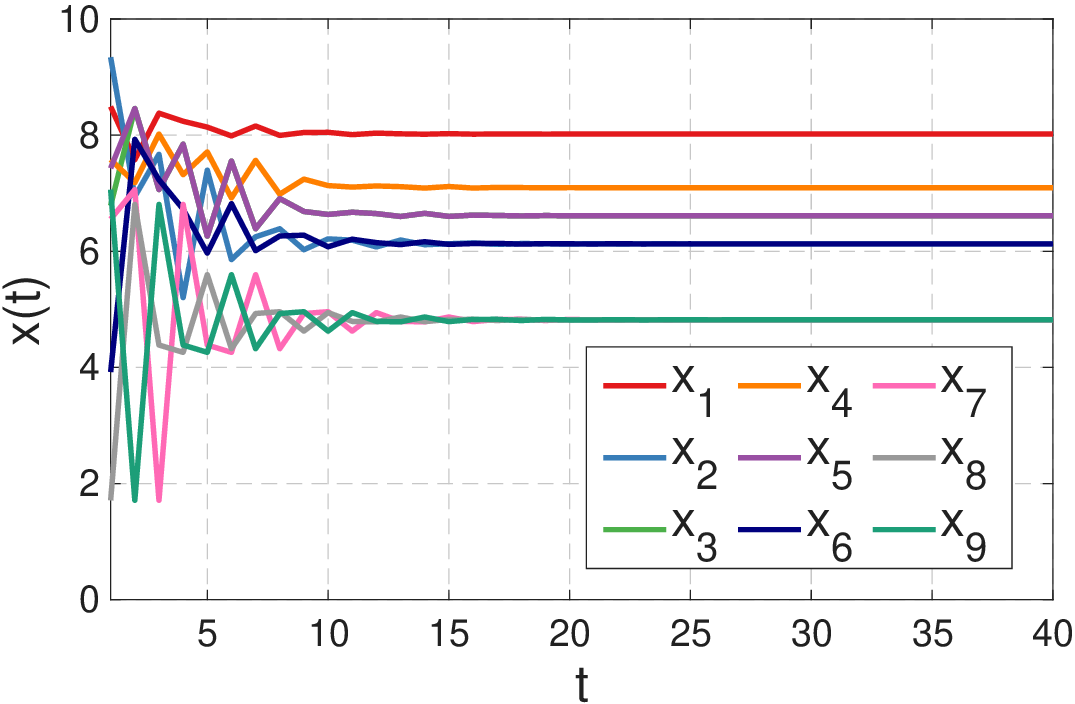}
    \caption{Formation of opinions clusters}
    \label{fig:plot_opinion_clusters}
    \end{subfigure}
    \caption{This figure illustrates the formation of opinion clusters under the FJ model due to the topological properties of $\mathcal{G}$.}
\end{figure}

Theorem \ref{thm:WC} significantly generalises the LTP-based conditions presented in \cite{shrinate2025opinionclusteringfriedkinjohnsenmodel} (Lemma \ref{lemma:ACC}) 
by showing that the LTP-based conditions are both necessary and sufficient when the agents forming the opinion cluster have a stubborn agent. Moreover, for a set of agents without any stubborn agents, Theorem \ref{thm:WC} presents a more generalised necessary and sufficient condition in eqn. \eqref{eqn:rank_condition} with the LTP-based condition as its special case. 

{Using the results established in Theorem \ref{thm:WC}, we present a partition of  the agents in a network $\mathcal{G}$ into multiple subgroups such that each subgroup forms an opinion cluster and overall the  
network achieves multiconsensus under the FJ model.} 

\begin{cor}
\label{cor:multi_consensus}
Consider a network $\mathcal{G}=(\mathcal{V},\mathcal{E})$ with $m\geq 1$ stubborn agents and $u\geq 0$ oblivious agents such that Assumption \ref{assump:ns} holds. Let the opinions of agents be governed by the FJ model \eqref{eq:FJ_opinion_dynamics}. Let $\mathcal{O}$ be a set containing the agents belonging to each of the oblivious iSCCs. 
The network \eqref{eq:FJ_opinion_dynamics} achieves multiconsensus with atmost $u+k$ opinion clusters if the agents $\mathcal{V}\setminus\mathcal{O}$ can be partitioned into $k$ disjoint subgroups such that each subgroup without any stubborn agents satisfies condition C1.) and each subgroup with a stubborn agent satisfies condition C2).    
\end{cor}

 %\end{cor}

The proof follows directly from Theorem \ref{thm:WC}, hence, it is omitted. Corollary \ref{cor:multi_consensus} presents conditions under which the agents in a weakly connected digraph achieve multiconsensus under the FJ model. Under the given conditions, the agents belonging to each of the oblivious iSCCs form $u$ disjoint subgroups and the remaining agents $\mathcal{V}\setminus \mathcal{O}$ form another $k$ disjoint subgroups. Thus, a digraph satisfying the conditions presented in Corollary \ref{cor:multi_consensus} achieves multiconsensus with atmost $k+u$ opinion clusters. %Each of these subgroups contains the oblivious agent(s) in an iSCC and the oblivious agent(s) that has a path(s) only from this iSCC. Further, the remaining agents $\mathcal{V}\setminus \mathcal{O}$ form another $k$ disjoint subgraph.  
\vspace{-10pt}
\subsection{Conditions for Cluster Consensus}
%\end{comment}

%The proof of this result is presented in the Appendix.

 %On the contrary, Theorem \ref{thm:SN_MC} shows that 

% In the following result, we present such a partition of the network where the agents in two or more different subgroups do not form an opinion cluster.
Under multiconsensus, the agents belonging to two distinct subgroups of the network can converge to the same final opinion. However, under cluster consensus, the agents in different subgroups must always converge to distinct final opinions.     {Therefore, the emergent behaviour of agents is more well-defined under cluster consensus than multiconsensus. Moreover, cluster consensus is also better suited to represent the echo chambers, where there is persistent disagreement among individuals with distinct ideologies.} In this subsection, we derive the conditions under which the network achieves cluster consensus under the FJ model.

%In this section, we present a partition of the network to achieve   thereby generalising the conditions for cluster consensus in \cite{yao2022cluster} to digraphs with cycles formed by multiple stubborn agents.
%\begin{remark}
%Under multiconsensus, the agents in a network belonging to two different subgroups can merge to form an opinion cluster. On the contrary,
Consider a set $\mathcal{C}$ such that $\mathcal{C}=\{p\}\cup \mathcal{N}_p$, where $p$ is a stubborn LTP agent.
By condition C2.) in Theorem \ref{thm:WC}, any superset of $\mathcal{C}$ cannot form an opinion cluster. Now, consider a network $\mathcal{G}$ with $m$ stubborn agents that is partitioned into disjoint subgroups $\mathcal{V}_1,...,\mathcal{V}_m$ such that each subgroup $\mathcal{V}_i$ contains a stubborn agent LTP agent $s_i$ and those persuaded by it $\mathcal{N}_{s_i}$ for $i\in[m]$. By Theorem \ref{thm:WC}, the agents belonging to any two distinct subgroups $\mathcal{V}_i$ and $\mathcal{V}_j$ cannot merge to form an opinion cluster (have consensus for all initial states).
However, Theorem \ref{thm:WC} does not rule out that for some special initial states, the final opinions of agents in $\mathcal{V}_i$ and $\mathcal{V}_j$ can be equal. The following result shows that these initial opinions form a measure-zero set.
%\end{remark}

\begin{comment}
We demonstrate this using the following example:
\begin{expm}
 Consider the network $\mathcal{G}$ with stubborn agents $2,4$ and $6$ shown in Fig. \ref{}. Each of the stubborn agents are LTP agents such that $\mathcal{N}_2=\{3\}, \mathcal{N}_4=\{4\}$ and $\mathcal{N}_6=\{1\}$. Hence, the network can be partitioned such that each subgroup contains an LTP agent and the nodes persuaded by it. 
 Under the FJ model, the following initial conditions lead to two subgroups forming the same final opinion
 \begin{enumerate}
     \item $\mathbf{x}_0=[] $
 \end{enumerate}
\end{expm}
    
\end{comment}

\begin{thm}
\label{thm:CC}
    Consider a network $\mathcal{G}$ with $m\geq 1$ stubborn agents such that the Assumption \ref{assump:ns} holds. The opinions of agents in $\mathcal{G}$ are governed by the FJ model \eqref{eq:FJ_opinion_dynamics}. The system \eqref{eq:FJ_opinion_dynamics} achieves cluster consensus with $m$ opinion clusters for \textit{almost all} initial opinions if and only if $\mathcal{G}$ can be partitioned into $m$ subgroups $\mathcal{V}_1,...,\mathcal{V}_m$ such that each subgroup $\mathcal{V}_i$ is composed of a stubborn LTP agent $s_i$ and those persuaded by it $\mathcal{N}_{s_i}$ for $i\in[m]$.
\end{thm}

The proof of Theorem \ref{thm:CC} is given in the Appendix. Theorem \ref{thm:CC} presents a necessary and sufficient condition for cluster consensus in the FJ framework. 
This result also generalises the conditions for cluster consensus presented in \cite{yao2022cluster} to digraphs where a cycle can contain multiple stubborn agents. Moreover, since the final opinions of the oblivious and non-oblivious agents cannot be equal under the FJ model (Lemma 4 in \cite{yao2022cluster}). Theorem \ref{thm:CC} can be extended to include networks with oblivious agents.
%An LTP agent is defined based only on paths from the influential agents; consequently, the final opinions of the agents in set $\mathcal{N}_p \cup \{p\}$ are equal, independent of the edge weights of the network and the stubbornness values. Thus, a major advantage of Theorem \ref{thm:WC} is that it presents a purely graph topology-based sufficient condition that ensures a set of agents form an opinion cluster in the FJ framework. 

\begin{cor}
\label{cor:CC}
    Consider a network $\mathcal{G}$ with $m\geq 1$ stubborn agents and oblivious agents such that Assumption \ref{assump:ns} holds. The opinions of agents in $\mathcal{G}$ are governed by the FJ model \eqref{eq:FJ_opinion_dynamics}.
     The network \eqref{eq:FJ_opinion_dynamics} achieves cluster consensus with $m+1$ opinion clusters for \textit{almost all} initial states if and only if the following conditions hold:
      \begin{itemize}
         \item the non-oblivious agents in $\mathcal{G}$ can be partitioned into $m$ subgroups such that each subgroup is composed of a stubborn agent LTP agent $s_i$ and those persuaded by it $\mathcal{N}_{{s}_i}$ for $i\in[m]$,

\item the oblivious agents form a unique aperiodic iSCC.
     \end{itemize}
    
\end{cor}

%topology-based conditions for cluster consensus in digraphs that do not satisfy the assumption that each 

\begin{comment}
\begin{thm}
\label{thm-SC}
    Consider a network $\mathcal{G}$ of $n$ agents with $m$ stubborn agents and no oblivious agents. The opinions of agents in $\mathcal{G}$ are governed by the FJ model \eqref{eq:FJ_opinion_dynamics}. If $p$ is an LTP agent, then
    opinions of $p$ and nodes in set $\mathcal{N}_p$ converge to the same final opinion and form an opinion cluster.
\end{thm}    
\end{comment}

%Under the FJ model, the opinions converge in the convex hull of the initial opinions of the initial agents. 

% 

%Using Theorem \ref{thm:WC}, the following result presents the topological conditions that result in multiconsensus. %with each opinion cluster formed by a desired set of agents.

%The proof of Corollary \ref{cor:cluster_consensus} is given in the Appendix.  The following example illustrates the conditions presented in Corollary \ref{cor:cluster_consensus} to obtain the desired opinion clustering.

% follows from Theorem \ref{thm:WC} and is omitted for brevity. (See the arXiv version for detailed proof and simulations \cite{}.)

% Next, we present the time complexity of identifying the LTP agents.

 % to assess the feasibility of identifying them. 
 The proof follows from Theorem \ref{thm:CC} and is omitted for brevity.
%An LTP agent is defined based only on paths from the influential agents. %; consequently, the final opinions of the agents in set $\mathcal{N}_p \cup \{p\}$ are equal, independent of the .
A key advantage of the LTP-based conditions proposed in Theorem \ref{thm:CC} is that it presents a purely graph topology-based and does not depend on the edge weights of the network and the stubbornness values.

  \begin{remark}
\textit{(Applicability to system beyond social networks) }Theorem \ref{thm:CC} highlights that the FJ model can be employed as a distributed control law for achieving cluster consensus in systems beyond social networks as well. To list a few, cluster consensus is important in multi-robot systems, power grids, oscillators, \textit{etc.} \textit{Why FJ?} Most of the existing protocols for cluster consensus \cite{xia2011clustering,luo2021cluster,QIN20132898}, along with topological conditions, also impose constraints on the edge weights. For instance, the pinning control-based techniques \cite{xia2011clustering,luo2021cluster,QIN20132898} consider the in-degree balance condition. Similarly, in \cite{maria}, each agent forming a cluster distributes a predefined positive weight among the cooperative interactions and a negative weight(s) among antagonistic interactions. In contrast: (a) the FJ model can ensure cluster consensus under purely topology-based conditions, independent of the edge-weights and stubbornness values. This is useful since the exact values of the edge weights (and stubbornness) may not be known in real-world networks. (b) As a result, this leads to a simpler network design process.  
\end{remark}

Next, we assess the  practical applicability of the proposed conditions for cluster consensus through a time complexity analysis. % of identifying the LTP agents. 

%The following remark 
%highlight 
\vspace{-10pt}
\subsection{Time complexity of identifying the LTP agents in a network}

%Identifying an agent $p$ is an LTP agent requires  
%In this subsection, we derive the time complexity of identifying all of the LTP agents in a network.
 Consider a network $\mathcal{G}$ with $m$ stubborn agents and $u$ oblivious iSCCs. To identify the LTP agents, we perform the following steps:
 \begin{itemize}[leftmargin=0pt, itemindent=1.5em]
     \item Consider an agent $p\in \mathcal{V}$ and remove all the outgoing edges of $p$, which takes $O(n)$ time.
     \item Next, the modified $\mathcal{G}$ is traversed $m+t$ times using the Depth First Search (DFS) algorithm, once from each of the $m$ stubborn agents and the $u$ oblivious iSCCs. Since each iteration of DFS has complexity $O(n+|\mathcal{E}|)$, this step requires $O((m+t)(n+|\mathcal{E}|))$ time.%, where $|\mathcal{E}|$ denotes the number of edges.
     \item If a node remains unvisited after all these traversals, then $p$ satisfies Condition (i) in Defn \ref{defn:LTP}. Let $\mathcal{U}_p$ be the set of unvisited nodes. By repeating this step for all $p\in \mathcal{V}$, we identify all agents that satisfy Condition (i) and store them as $\mathcal{H}$. This step has a time complexity of $O((u+m)n(n+|\mathcal{E}|))$ which reduces to $O(n(n+|\mathcal{E}|))$ since $u+m\ll n$. If $\mathcal{H}\subseteq \mathcal{I}$, identifying all LTP agents requires $O(n(n+|\mathcal{E}|))$ time.
     \item {Next, we verify condition (ii) in Defn. \ref{defn:LTP} for each non-influential agent in $\mathcal{H}$. For a non-influential agent $p$, we check if there is a $c\in \mathcal{H}\setminus\{p\}$ such that  $p\in\mathcal{U}_c$. This step requires $O(n^2)$ time. Repeating it for all $p\in \mathcal{H}$ yields a worst-case complexity of $O(n^3)$. Thus, the time complexity of identifying all LTP agents in $\mathcal{G}$ is $O(n^3)$.}
 \end{itemize}
 
To verify if $\mathcal{G}$ satsifies
the conditions in Theorem \ref{thm:CC} (or Corollary \ref{cor:CC}), we check if $\mathcal{H}$ is equal to the set of stubborn agents and $\cup_{p\in \mathcal{H}}\{p\} \cup \mathcal{U}_p=\mathcal{V}$ (or $\cup_{p\in \mathcal{H}}\{p\} \cup \mathcal{U}_p$ equals the non-oblivious agents in $\mathcal{G}$). Hence, the overall complexity of verifying whether a network achieves cluster consensus is $O(n(n+|\mathcal{E}|))$.

 \vspace{-10pt}
 \subsection{Discussion}
%Empirical evidence shows that although only a limited number of social media users are in echo chambers, their contribution to the discourse on social media is significant \cite{pratelli2024entropy}.  %Each discursive community consists of some influential agents and their followers. For instance: some verified and non-verified users on X.

A group of users with aligned views that forms an echo chamber on social media is also known as a discursive community \cite{pratelli2024entropy}. In \cite{mattei2022bow}, the authors construct a network from an X dataset using user interactions (retweets) on a specific topic. Thereafter, using a community detection algorithm, they identify the discursive communities therein. When the discussion topic is polarising (such as COVID-19, elections, \textit{etc.}), the authors show that each discursive community exhibits a bow-tie structure, shown in Fig. \ref{fig:Bow-tie}. 

A bow-tie structure can be decomposed into seven components: IN, SCC, OUT, INTENDRILS, TUBES, OUTTENDRILS, and OTHERS.
In a discursive community, the IN component, generally composed of verified X users  (such as journalists, celebrities, politicians, \textit{etc}), who post tweets and influence other users. The users in INTENDRILS only follow the tweets of those in IN but are themselves not followed by anyone. The users in SCC form the largest strongly connected component in the digraph, whereas those in TUBES form a weakly connected component. All of them follow the users in IN. %and are followed by those in OUT. The users in OUT follow those in SCC, \textcolor{red}{IN}, TUBES and OUTTENDRILS. 

%In a bow-tie structure, a network can be decomposed into seven components: IN, SCC, OUT, INTENDRILS, TUBES, OUTTENDRILS, and OTHERS. The IN component, generally composed of verified X users  (such as journalists, celebrities, politicians, \textit{etc.}), who post tweets and influence other users. The users in INTENDRILS only follow the tweets of those in IN but are themselves not followed by anyone. The users in SCC form the largest strongly connected component in the network and those in TUBES are weakly connected. All of them follow the users in IN. %and are followed by those in OUT.

In prominent bow-tie structures, the components OTHERS and OUTENDRILS consist of very few users and can be neglected. In the  resulting bow-tie structure, the IN component acts like a stubborn LTP agent (composed of ideologically aligned verified users) and the remaining components (\textit{i.e.} INTENDRILS, SCC, OUT, TUBES) are persuaded by it. Therefore, the empirical study in \cite{mattei2022bow} shows that the topological interconnections between an LTP and its persuaded agents emerge naturally in several real-world echo chambers. 

%the topological property that is shown to result in echo chamber formation under the FJ model in Theorem \ref{thm:CC} is also seen in the real world echo chambers.

%Therefore, the empirical study in \cite{mattei2022bow} shows that the topological interconnections among the users forming echo chambers in several real-world  networks naturally resemble the LTP and its persuaded agents' structure.

%`the topological relation among an emerges naturally in several real-world echo chambers. 

%This emperical study shows that with under the formation of echo chamberers, each subgroup consists of unique  stubborn LTP agent and the remaining are persuaded by it. 

%However, this bow-tie architecture was not prominent found in the discursive communities discussing non-polarising issues (such as EURO Cup 2020).

%Bow tie structures are found in several complex networks such as web networks, gene regulator networks \textit{etc.}.  

%consists of $7$ nodes,  shown in Fig. \ref{}, the IN node consists generally of the verified users that post content, which is further amplified by those users forming the strongly connected component (SCC) which are later followed by OUT nodes.

\begin{figure}[h]
    \centering
    \includegraphics[width=0.5\linewidth]{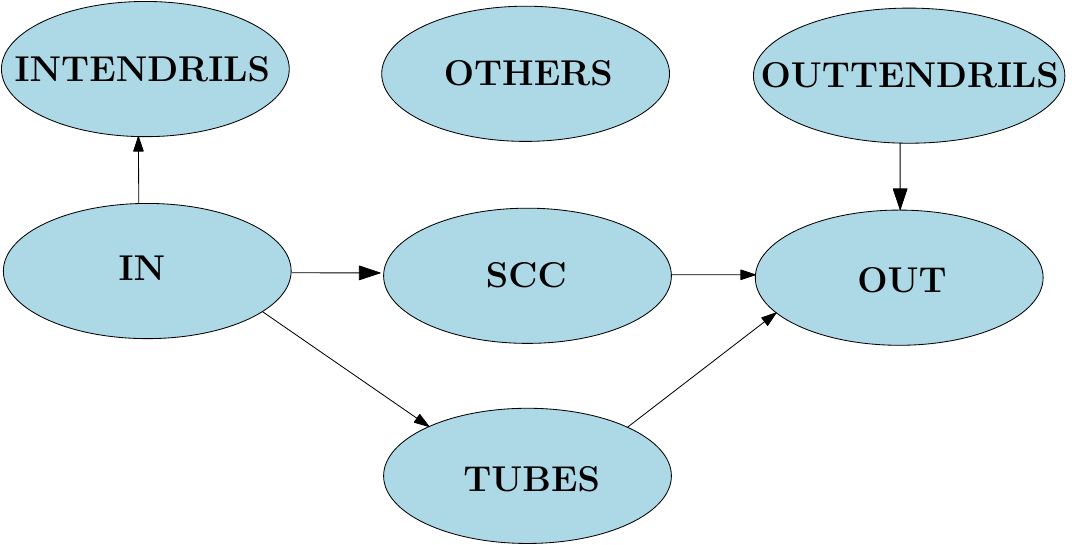}
    \caption{Bow-tie structure}
    \label{fig:Bow-tie}
\end{figure}
%\vspace{-10pt}
%For each data set, based on interactions (measured by retweets) the authors constructed an underlying interaction network of both verified and non-verified users and used a community detection protcol to segregrate them into possible ideological or political factions. Interestingly, the discursive community of each faction on polarising issues formed a bow-tie architecture, however this artictecture was not prominent in non-polarising issues.  
\vspace{-5pt}
\section{Conclusions} 
In this paper, we examine the cluster consensus problem in the FJ framework to explain the role of interactions in the emergence of echo chambers in social networks. First, we present the necessary and sufficient conditions under which any set of agents forms an opinion cluster under the FJ model. The proposed conditions depend on whether or not the set contains a stubborn agent. If a stubborn agent exists, then it must be an LTP agent and the rest in the set must be persuaded by it. For a set of non-stubborn agents, we obtain a more generalised \textit{if and only if} condition eqn. \ref{eqn:rank_condition} that leads to an opinion cluster with the LTP-based condition being its special case. Notably, under the generalised condition, the edge weights must satisfy eqn.  \eqref{eqn:rank_condition}, while the LTP-based condition does not depend on the edge weights and stubbornness values.

%unlike the LTP-based condition.%if and only if the stubborn agent is a 
%This work examines the conditions under which the agents that lead to role of network topology in the formation of opinion clusters in the FJ framework. 
%In this work, we examine the role of network interactions in the emergence of cluster consensus. Using these results, we explain with the objective of explaining the role of 
%formation of opinion clusters in the FJ frameworkwe examine the role of network interactions and present the necessary and sufficient conditions for a set of agents to form opinion clusters under the FJ model. Interestingly, if the set contains a stubborn agent, it forms an opinion cluster if and only if the stubborn agent is an LTP and the rest are persuaded by it. However, for a set without any stubborn agents, opinion clusters can form for some special values of edge weights, even when none of the agents is an LTP agent. 
%Here, we further explored the abilities of the LTP agents, introduced in \cite{shrinate2025opinionclusteringfriedkinjohnsenmodel}.  

Building on these results, we show that a network with $m$ stubborn agents achieves cluster consensus with $m$ opinion clusters \textit{if and only if} it can be partitioned into $m$ disjoint subgroups, each consisting of a stubborn LTP agent and those persuaded by it. The proposed condition for cluster consensus generalises \cite{yao2022cluster} to a broader class of digraphs where multiple stubborn agents can exist in a cycle. Moreover, it provides a theoretical explanation for the prominent bow-tie structure, that characteristically emerges in several real-world echo-chambers.
Since the proposed conditions are topology-based, networks satisfying these conditions are simple to design for applications in engineered networks. To this end, we also present a computationally efficient methodology with a time complexity of $O(n(n+|\mathcal{E}|))$ that verifies whether a network satisfies the proposed conditions. In future, it will be interesting to develop algorithms that suitably modify the network such that cluster consensus with any desired partitioning of agents is achieved. %Additionally, since the proposed conditions depend on network topology, the FJ model can be used for applications in engineered networks. Moreover, Finally, we present computationally efficient algorithm with a time complexity of $O(n(n+|\mathcal{E}|))$ to verify whether a network satisfies the proposed conditions. 

\label{Sec:conc}
\begin{comment}
by the
, there is a path from .

agent  because by the  Since $5$ is not an LTP and is not persuaded by another LTP agent, it forms a .

as demonstrated in Example \ref{Example-1},
However, if an  Similarly, topological modifications can result in a node

However,  The following example shows that the opinion clusters that form in the final opinion are independent of the degree of stubbornness of the existing stubborn agents. % in network $\math$
 \begin{expm}
    Consider the network $\mathcal{G}$ shown in Fig. \ref{Fig:expm_1_graph}.   Now, we modify the degree of stubbornness $\beta_2$ from $0.1$ to $0.99$ and plot the final opinions in Fig. \ref{Fig:final_op_and_stubbornnesn+2}. Independent of the magnitude of $\beta_2$, three clusters form in the final opinion with the same set of agents. A similar effect is observed by modifying $\beta_6$ as shown in Fig. \ref{Fig:final_op_and_stubbornness}.
\end{expm}
\begin{figure}[h]
    \centering
    \begin{subfigure}{0.23\textwidth}
        \centering
    \includegraphics[width=1\linewidth]{influence_3.pdf}
    \caption{Final opinions under varying $\beta_2$}
    \label{Fig:final_op_and_stubbornnesn+2}
    \end{subfigure}
    \hfill
        \begin{subfigure}{0.23\textwidth}
        \centering
    \includegraphics[width=1\linewidth]{influence_.pdf}
    \caption{Final opinions under varying $\beta_6$}
       \label{Fig:final_op_and_stubbornness}
    \end{subfigure}
    \caption{Change in final opinions under varying $\beta_2$ and $\beta_6$.}
    \label{fig:placeholder1}
\end{figure}   
\end{comment}
 %\vspace{-10pt}
 \vspace{-5pt}
\section*{Appendix}

\textbf{\textcolor{cyan}{Proof of Theorem-\ref{thm:WC}}}: 
At steady state, we can write the eqns. in \eqref{eqn:2_cluster_ss} as follows:
\begingroup
\setlength{\abovedisplayskip}{2pt}
\setlength{\belowdisplayskip}{2pt}
\begin{align}
   \label{eqn:R_matrix}
  & \tiny
  \underbrace{ \left[\begin{array}{c:c:ccc}
     I-(I-{\beta}_{11}) \hat{W}_{11} & -(I-\beta_{11}) \hat{W}_{12} &
     ... & -(I-\beta_{11}) \hat{W}_{1(u+2)} & -\phi \\ \hdashline
              -(1-\beta_{22}) \hat{W}_{21} & I-(I-\beta_{22})\hat{W}_{22}& 
              ... & -(1-\beta_{22}) \hat{W}_{2(u+2)} & -\psi \\ \hdashline
              \vdots & \vdots & \ddots & \vdots 
              & \vdots \\
              \mathbb{0} & \mathbb{0} & \mathbb{0}  &
              I-\hat{W}_{(u+2)(u+2)} & \mathbb{0} \\
                \mathbb{0} & \mathbb{0} & \mathbb{0} & \mathbb{0} & \mathbb{0}
  \end{array}
\right]}_{R} . \nonumber\\
& \scriptsize
 \qquad  \qquad  \qquad \qquad \qquad \qquad  \begin{bmatrix}
        \hat{\mathbf{x}}^*_1 &\hat{\mathbf{x}}^*_2 & \cdots & \hat{\mathbf{x}}_{u+2}^* & \mathbf{x}_s(0)
    \end{bmatrix}^T =\mathbb{0}.  
\end{align}
\endgroup    

 %\textcolor{magenta}{Reduce the font size of eqn. 7 slightly, so that even the equation number fits. We will see how it looks.} 
Here, $\hat{\mathbf{x}}_i^*:=\lim_{k \to \infty} \hat{\mathbf{x}}_i(k)$ for $i\in\{1,...,u+2\}$, $\mathbf{x}_s=[x_{s_1}(0),...,x_{s_m}(0)]$ contains the initial opinions of the $m$ stubborn agents. %(, labelled $s_1,...,s_m$. 
{The matrix $\phi=[\phi_{ij}]$ is defined as: $\phi_{ij}=\beta_i$  if stubborn agent $i\in \{1,...,|\mathcal{C}|\}$ is labelled $s_j$; otherwise $\phi_{ij}=0$.} %Similarly, matrix $\psi=[\psi_{ij}]$ is defined as: $\psi_{ij}=\beta_i$ if stubborn agent $i\notin\mathcal{C}$ is labelled $s_j$, otherwise  $\psi_{ij}=0$.}

By construction, the matrix $R$, defined in eqn. \eqref{eqn:R_matrix}, is a Laplacian matrix. The associated digraph $\mathcal{G}(R)$ contains $n+m$ nodes and has an edge $(i,j)$ if $[R]_{ji}=r_{ji}\neq 0$. We index the nodes in $\mathcal{G}(R)$ such that a node $i \in \{1,...,n\}$ is associated with the final opinion of the agent $i$ (consistent with the indexing in $\mathcal{G}$) and a node $n+i,~i\in [m]$ is associated with the initial opinion of the stubborn agent $s_i$. Each stubborn agent contributes two nodes in $\mathcal{G}(R)$; one associated with its initial opinion and another associated with its final opinion. The nodes $n+1,...,n+m$ form sources in $\mathcal{G}(R)$. Note that each source $n+i$ has a single outgoing edge only to the node associated with the final opinion of the corresponding stubborn agent. 

{Since $R$ is a Laplacian matrix, we can employ Kron Reduction to reduce the associated digraph $\mathcal{G}(R)$. In a digraph $\mathcal{G}$, each non-stubborn agent, which is not in any oblivious iSCCs, has atleast one path from a stubborn agent or an agent in an oblivious iSCC. Consequently, each node in $\mathcal{G}(R)$  indexed $\{1,...,n\}$, that is not associated with an agent in any oblivious iSCC, has a path in $\mathcal{G}(R)$ from a node in $\{n+1,...,n+m\}$ or a node associated with an oblivious iSCC. %Since $\alpha$ contains the nodes $\{n+1,...,n+m\}$ and the agents in all the oblivious iSCCs,
Hence, for a set $\alpha \subset [n+m]$ that contains the nodes $\{n+1,...,n+m\}$ and all the nodes associated with the oblivious iSCCs, the Kron-Reduced matrix $R/\alpha^c$ is well defined (\textit{i.e.} $R[\alpha^c]=I-(I-\beta_{22})W_{22}$ is invertible)} (see Lemma 3.3 in \cite{Kron_red_digraphs}).
Next, using the Kron reduction of $\mathcal{G}(R)$, we present the proof of condition C1.)

%(\textcolor{magenta}{Why do we need to add it here?} Refer to \cite{shrinate2025opinionclusteringfriedkinjohnsenmodel} for detailed construction of $\mathcal{G}(R)$).

\textbf{Proof of C1:}
Just to recall, $\mathcal{C}$ consists only of non-stubborn agents who also do not belong to any oblivious iSCC. %We then construct a set
Consider a set $\alpha$ comprising all the nodes in $\mathcal{C}$, the agents in each of the $u$ oblivious iSCCs and the nodes $\{n+1,...,n+m\}$. %In any weakly connected digraph $\mathcal{G}$, each non-stubborn agent, not in any oblivious iSCC, has atleast one path either from a stubborn agent or an agent in an oblivious iSCC. Therefore, by construction of $\mathcal{G}(R)$, each node associated with a non-stubborn agent has a path in $\mathcal{G}(R)$ from a node in $\{n+1,...,n+m\}$ or a node associated with an oblivious iSCC. %Since $\alpha$ contains the nodes $\{n+1,...,n+m\}$ and the agents in all the oblivious iSCCs,Hence, each node in $\alpha^c$ has a path in $\mathcal{G}(R)$ from a node in $\alpha^c$. 
Since $\alpha$ contains the nodes $\{n+1,...,n+m\}$ and the agents in all the oblivious iSCCs, $R/\alpha^c$ is well defined. By the definition of $\alpha$, $R/\alpha^c$ is the Schur complement of  $I-(I-\beta_{22})W_{22}$ in $R$ (as highlighted in \eqref{eqn:R_matrix}) and can be evaluated using eqn. \eqref{eqn:Schur_complement}. Consequently, $W_{22}$ has order $|\alpha^c| \times |\alpha^c|$. %\textcolor{magenta}{cite the other paper on Kron-reduction for digraphs and explain the idea briefly}. %By the definition of $\alpha$, $R/\alpha^c$ is the Schur complement of  $I-(I-\beta_{22})W_{22}$ in $R$ and can be evaluated using eqn. \eqref{eqn:Schur_complement}. (Consequently, $W_{22}$ has order $|\alpha^c| \times |\alpha^c|$.)

Next, we partition the obtained $R/\alpha^c$ into submatrices as follows:
\begingroup
\setlength{\abovedisplayskip}{2pt}
\setlength{\belowdisplayskip}{2pt}
\begin{align*}
\scriptsize
R/\alpha^c=\begin{bmatrix}
   ^{\alpha}R_{11} & ^{\alpha}R_{12} & ... & ^{\alpha}R_{1(u+1)} & ^{\alpha}R_{1(u+2)} \\
   \mathbb{0} & ^{\alpha}R_{22} & ... & \mathbb{0} & \mathbb{0} \\
   \vdots & \cdots & \ddots & \cdots & \vdots \\
  \mathbb{0} & \mathbb{0} & \mathbb{0} & ^{\alpha}R_{(u+1)(u+1)} & \mathbb{0} \\
  \mathbb{0} & \mathbb{0} & \mathbb{0} & \mathbb{0} & \mathbb{0}
\end{bmatrix}
\end{align*}
\endgroup
The submatrices $^{\alpha}R_{ij}$ are defined as: 
\begin{enumerate}
    \item $^{\alpha}R_{11}=I-\hat{W}_{11}-Z(I-\beta_{22})\hat{W}_{21}$,
    \item  for $j\in\{2,...,u+1\}$, $^{\alpha}R_{1j}=-\hat{W}_{1(j+1)}-Z(I-\beta_{22})\hat{W}_{2(j+1)}$ and $^{\alpha}R_{jj}=I-\hat{W}_{(j+1)(j+1)}$,
    \item $^{\alpha}R_{1(u+2)}=-Z\psi$
\end{enumerate}
Since $R$ is a loopless Laplacian matrix, $R/\alpha^c$ is also a loopless Laplacian matrix by Statement 2) in Lemma \ref{lm:Basic_properties}.
%First, we show that $^{\alpha}R_{11}$ is an invertible matrix. 
Further, we know that each node in  $\mathcal{C}$ either has a path in $\mathcal{G}(R)$ from a node in $\{n+1,...,n+m\}$ or a node associated with an agent in an oblivious iSCC. By Statement 3) in Lemma \ref{lm:Basic_properties}, it follows that in the reduced digraph $\mathcal{G}_{\alpha}$ (associated with $R/\alpha^c$) as well, each node associated with $\mathcal{C}$ has a path from either nodes $\{n+1,...,n+m\}$ or the oblivious agents in the iSCCs. %As a result, for there is  the diagonal enteries of $^{\alpha}R_{11}$ is strictly greater that the sum of absolute values
{By Lemma 3.3 in \cite{Kron_red_digraphs}, the matrix $^{\alpha
}R_{11}$ is invertible.}

%by Statement 1 in Lemma \ref{lm:Basic_properties}, the Kron Reduction of $(R/\alpha^c)/\mathcal{C}$ is well defined. Thus, it follows that $^{\alpha}R_{11}$ is invertible.

{As discussed in Sec. \ref{subsec:KR}, evaluating the Schur complement $R/\alpha^c$ results in eqn. \eqref{eqn:R_matrix} being reduced to,}  
\begingroup
\setlength{\abovedisplayskip}{2pt}
\setlength{\belowdisplayskip}{2pt}
\begin{align}
\label{eqn:reduced_LEs}
    R/\alpha^c\cdot \begin{bmatrix}
        \hat{\mathbf{x}}_1^* & \hat{\mathbf{x}}_3^* & ...& \hat{\mathbf{x}}_{u+2}^*
    \end{bmatrix}= \mathbb{0} 
\end{align}
\endgroup
From eqn. \eqref{eqn:reduced_LEs}, we get $^{\alpha}R_{11}\hat{\mathbf{x}}_1^*+(^{\alpha}R_{12}\hat{\mathbf{x}}_3^*+...+^{\alpha}R_{1(u+1)}\hat{\mathbf{x}}_{u+2}^* $ $+^{\alpha}R_{1(u+2)}\mathbf{x}_s(0)=\mathbb{0}$. Since the oblivious agents in each of the iSCCs achieve consensus, $\mathbf{\hat{x}}_i^*=\mathbb{1} b_{i}$, where $b_{i}\in \mathbb{R}$ depends on the initial opinions and connectivity among the agents in the corresponding iSCC and $i\in\{3,...,u+2\}$. Thus,
\begingroup
\setlength{\abovedisplayskip}{2pt}
\setlength{\belowdisplayskip}{2pt}
\begin{align}
\label{eqn:impact_on_final_op}
 \hat{\mathbf{x}}_1^*=-(^{\alpha}R_{11})^{-1}\big(^{\alpha}R_{12}\mathbb{1}b_3+...+&^{\alpha}R_{1(u+1)}\mathbb{1}b_{u+2} +\nonumber\\
 &^{\alpha}R_{1(u+2)}\mathbf{x}_s(0)\big)   
\end{align}
\endgroup

\begin{comment}
\begin{align*}
\hat{\mathbf{x}}_1^*=-(^{\alpha}R_{11})^{-1}\big(
\begin{bmatrix}
 ^{\alpha}R_{12}\mathbb{1} & ... & ^{\alpha}R_{1(u+1)}\mathbb{1} & ^{\alpha}R_{1(u+2)} 
\end{bmatrix} \begin{bmatrix}
  b_3 \\ \vdots \\b_{u+2} \\ \mathbf{x}_s(0)
\end{bmatrix}
\big)    
\end{align*}    
\end{comment}
(\textbf{Neccesity}) The agents in $\mathcal{C}$ have consensus (satisfy eqn. \eqref{eqn:opinion_cluster}) for all initial opinions of agents if and only 
\begingroup
\setlength{\abovedisplayskip}{2pt}
\setlength{\belowdisplayskip}{2pt}
\begin{align*}
   (^{\alpha}R_{11})^{-1}\big(
\begin{bmatrix}
 ^{\alpha}R_{12}\mathbb{1} & ... & ^{\alpha}R_{1(u+1)}\mathbb{1} & ^{\alpha}R_{1(u+2)} 
\end{bmatrix}= \mathbb{1} \mathbf{f}^T 
\end{align*}
\endgroup
where $\mathbf{f}\in \mathbb{R}^{m+u}$.
Since $(^{\alpha}R_{11})^{-1}$ is invertible,
$\operatorname{rank}([
 ^{\alpha}R_{12}\mathbb{1} \ \ ... \ $ $ ^{\alpha}R_{1(u+1)}\mathbb{1}  \  \    ^{\alpha}R_{1(u+2)} 
])$ $=\operatorname{rank}(\mathbb{1} \mathbf{f}^T)=1$ (see 0.4.6(b) in \cite{horn2012matrix}). Thus, $\operatorname{rank}([
 ^{\alpha}R_{12}\mathbb{1} \ \ ... \ \  ^{\alpha}R_{1(u+1)}\mathbb{1}  \ \ ^{\alpha}R_{1(u+2)}])=1$ is the necessary condition for $\mathcal{C}$ to form an opinion cluster. Expressing the matrices $^\alpha R_{1j}$ in terms of  $\hat{W}_{ij},Z$ and $\psi$, yields eqn. \eqref{eqn:rank_condition}.

(\textbf{Sufficiency}) Eqn. \eqref{eqn:rank_condition} implies that rank of $[
 ^{\alpha}R_{12}\mathbb{1} \ \  ...  \ \ ^{\alpha}R_{1(u+1)}\mathbb{1} $ $ \ \ ^{\alpha}R_{1(u+2)}]$ is $1$ and it can be
expressed as $\mathbf{y}\Delta^T$ for some $\mathbf{y} \in \mathbb{R}^{|\mathcal{C}|}$ and $\Delta=[\delta_i] \in \mathbb{R}^{u+m}$. Thus, we can write $^{\alpha}R_{1j}\mathbb{1}=$ $\mathbf{y}\delta_{j-1}$ for $j\in\{2,...,u+1\}$ and $^{\alpha}R_{1(u+2)}=\mathbf{y}[\delta_{u+1},...,\delta_{u+m}]^T$. Since $R/\alpha^c$ is a loopless Laplacian matrix, $^{\alpha}R_{11} \mathbb{1}+^{\alpha}R_{12} \mathbb{1}+...+^{\alpha}R_{1(u+2)}\mathbb{1}=\mathbb{0}$. Hence, it is simple to see that $^{\alpha}R_{11}\mathbb{1}=-\mathbf{y}\zeta$, where $\zeta \in \mathbb{R}$ and, $\zeta(^\alpha R_{11})^{-1}\mathbf{y}=\mathbb{1}$. Substituiting $^{\alpha}R_{1j}$ in eqn. \eqref{eqn:impact_on_final_op}, we get:\begingroup
\setlength{\abovedisplayskip}{2pt}
\setlength{\belowdisplayskip}{2pt}
\begin{align*}
    \mathbf{\hat{x}}_1^*&=-(^{\alpha}R_{11})^{-1}\mathbf{y}\big(\delta_1b_3+...+\delta_ub_{u+2}+[\delta_{u+1},...,\delta_{u+m}]^T\mathbf{x}_s(0)\big)\\
    &=-\mathbb{1}\zeta_2, \qquad \textit{ where } \zeta_2\in\mathbb{R}.
\end{align*}
\endgroup
Hence, eqn. \eqref{eqn:rank_condition} is sufficient to ensure that $\mathcal{C}$ form an opinion cluster.

\textbf{Proof of C2:} \textbf{(Sufficiency):}
By Lemma \ref{lemma:ACC}, if  $\mathcal{C}$ is composed of a stubborn LTP agent $p$ and those in $\mathcal{N}_p$, it forms an opinion cluster. 
%The sufficiency condition follows directly . Hence, we prove the necessary condition.

\textbf{(Neccesity):}
Let  $\mathcal{C}$ contain atleast one stubborn agent and form an opinion cluster. %As discussed, the agents in $\mathcal{C}$ can form an opinion cluster only if each influential agent equally impacts the agents in $\mathcal{C}$. 
To determine the conditions under which this holds, we again consider $\alpha$ to be the union of all nodes in $\mathcal{C}$, the agents in the oblivious iSCCs and $\{n+1,...,n+m\}$. Further, we evaluate $R/\alpha^c$ using eqn. \eqref{eqn:Schur_complement} and obtain the reduced linear eqns. \eqref{eqn:reduced_LEs}.
We first establish that $\mathcal{C}$ can contain atmost one stubborn agent. On the contrary, let $\mathcal{C}$ contain two stubborn agents such that they are indexed $1$ and $2$. The nodes in $\mathcal{G}(R)$ corresponding to their initial opinions are indexed $n+1$ and $n+2$, respectively. By Statement 3 in Lemma \ref{lm:Basic_properties}, for $i,j\in \alpha$, an edge $(i,j)$ exists in $\mathcal{G}_{\alpha}$ if and only if there is a path in $\mathcal{G}(R)$ from a node $i$ to $j$ that is composed only of nodes in $\alpha^c,i$ and $j$. Equivalently, the corresponding entry of $R/\alpha^c$ is non-zero. By construction of $\mathcal{G}(R)$, $n+1$ only has an outgoing edge to $1$ and $n+2$ only has an outgoing edge to $2$. Hence, each directed path from $n+1$ (resp. $n+2$) to any node in $\mathcal{G}(R)$ traverses $1$ (resp. $2$). 
Since $1$ and $2$ lie in $\alpha$, hence,
%know that for any $d,e\in \alpha$, the corresponding entry $r_{F(d),F(e)}^1$ of  $R/\alpha^c$ is non-zero only when there is an edge $(e,d)$ or a path $e \to d$  in $\mathcal{G}(R)$ that traverses only the nodes in $\alpha^c$ (except $e$ and $d$). Thus,  it follows that:
\begingroup
\setlength{\abovedisplayskip}{2pt}
\setlength{\belowdisplayskip}{2pt}
\begin{align*}
^{\alpha}R_{1,(u+2)}[1:|\mathcal{C}|,1]=
%R/\alpha^c[F(\mathcal{C}),F(n+1)]=
[\beta_1,\mathbb{0}_{|\mathcal{C}|-1}]^T \\  
^{\alpha}R_{1,(u+2)}[1:|\mathcal{C}|,2]=[0,\beta_2,\mathbb{0}_{|\mathcal{C}|-2}]^T. 
\end{align*}
\endgroup

Next, let $H=[h_{ij}]=(^{\alpha}R_{11})^{-1}$. %Then, we can rewrite eqn. \eqref{eqn:eqn-1} as, %it follows that the contribution of the stubborn agents $i$ and $j$
From eqn. \eqref{eqn:impact_on_final_op}, the contribution of the initial opinion of stubborn agent $1$ in the final opinions of agents in $\mathcal{C}$ is given by $-H ^{\alpha}R_{1,(u+2)}[1:|\mathcal{C}|,1]x_1(0)=-[h_{11}\beta_1,h_{21}\beta_1,...,h_{|\mathcal{C}|1}\beta_1]^Tx_1(0)$. Equivalently, the contribution of the initial opinion of $2$ is given by $H^{\alpha}R_{1,(u+2)}[1:|\mathcal{C}|,2]x_2(0)=[h_{12}\beta_2,h_{22}\beta_2...,h_{|\mathcal{C}|2}\beta_2]^Tx_2(0)$. Since the agents in $\mathcal{C}$ form an opinion cluster, the impact of agent $1$ and $2$ on each agent in $\mathcal{C}$ must be equal. Consequently, $h_{11}=h_{21}=...=h_{|\mathcal{C}|1}$ and $h_{12}=h_{22}=...=h_{|\mathcal{C}|2}$. However, since $H$ is invertible, two of its columns cannot be linearly dependent. Therefore, a set $\mathcal{C}$ that forms an opinion cluster can have atmost one stubborn agent. 

From the above discussion, it follows that if $\mathcal{C}$ contains stubborn agent $1$, then $1$ is the unique stubborn agent in $\mathcal{C}$ and all the entries of the first column of $H$ are equal. %Let the first agent $i$ in $\mathcal{C}$ be the stubborn and the rest be non-stubborn. 
The matrix $H^{-1}=^{\alpha}R_{11}$ and it satisfies $H^{-1}H=I$. Consequently, $^{\alpha}R_{11}h_{11}\mathbb{1}=\mathbb{e}_1$. %Since the stubborn agent $1$ is the first node in $\mathcal{C}$, and all entries of the first column of $H$ are equal, 
\begin{comment}
the following eqn. holds:
\begingroup
\setlength{\abovedisplayskip}{2pt}
\setlength{\belowdisplayskip}{2pt}
\begin{align}
\label{eqn:row-sum}
%R/\alpha^c[F(\mathcal{C})]h_{11}\mathbb{1}_{|\mathcal{C}|}=e_1 

\end{align}
\endgroup    
\end{comment}
This equation implies that the row-sum of each row of $^{\alpha}R_{11}$ indexed $2,...,|\mathcal{C}|$ is zero. Since $R/\alpha^c$ is a Laplacian matrix, only its diagonal entries are non-negative, and the rest are non-positive. Consequently, only the first row of the submatrices $^{\alpha}R_{1j}$ for $j\in\{2,...,u+2\}$ can have non-zero entries and all entries in the remaining rows of these matrices are zero.
\begin{comment}
    
$R/\alpha^c[F(\mathcal{C}\setminus \{i\}),F([u])]$ and $R/\alpha^c[F(\mathcal{C} \setminus \{i\}),F(\mathcal{S})]$
\end{comment} 
This implies that each path in $\mathcal{G}$ from the oblivious agents in iSCCs and stubborn agents traverses $1$ to reach the remaining agents in $\mathcal{C}$. Since $1$ is stubborn and Condition (i) in Defn. \ref{defn:LTP} holds, thus $1$ is the LTP agent and the rest in $\mathcal{C}$ are in $\mathcal{N}_1$. \hfill$\blacksquare$

\textbf{\textcolor{cyan}{Proof of Theorem \ref{thm:CC}}:}
(\textbf{Sufficiency}) %The agents in $\mathcal{G}$ are partitioned into $m$ disjoint subgroups $\mathcal{V}_1,...,\mathcal{V}_{m}$ such that each group $\mathcal{V}_i$ is composed of a stubborn LTP agent $s_i$ and $\mathcal{N}_{s_i}$ for $i\in[m]$.
By Theorem \ref{thm:WC}, the following holds: (a) each subgroup $\mathcal{V}_i$ forms an opinion cluster (b) since each subgroup in $\mathcal{G}$ contains a stubborn agent, two distinct subgroups $\mathcal{V}_i$ and $\mathcal{V}_j$ cannot form one opinion cluster (\textit{i.e} they cannot converge to the same final opinion for all initial opinions), where $i,j\in\{1,...,m\}$).  %To prove that cluster consensus is achieved for almost all initial opinions, we examine if eqn. \eqref{eqn:distinct_op} holds for almost all initial opinions. Each . Thus, by Theorem \ref{thm:WC}, 
%each of these subgroups forms an opinion cluster if and only if the stubborn agent is an LTP agent and the remaining non-stubborn agents are persuaded by it. Hence, 
 However, it may occur that for some special initial opinions, the agents in distinct subgroups converge to the same final opinion value. Next, we show that such initial opinions form a measure-zero set in $\mathbb{R}^{n}$ and hence cluster consensus is achieved for almost all initial states.

%Since the oblivious agents are governed by the DeGroot's model, the agents in $\mathcal{V}_1$ form an opinion cluster (reach consensus for any initial opinions) as the oblivious agents contain a unique aperiodic iSCC \cite{degroot1974consensus}. Further,  A non-oblivious agent and an oblivious agent cannot have consensus (Lemma 4 in \cite{yao2022cluster}). Hence, the final opinions of agents in $\mathcal{V}_1$ and $\mathcal{V}_j$ are distinct. Further,

%we determine whether the final opinions of agents in distinct subgroup can converge for some special initial states. % can converge to the same varemain distinct Consequently, under the given partition of the agents, none of the two subgroups merge to form an opinion cluster Next, we show that the final opinions of any two opinion clusters formed by the stubborn agents, always converge to distinct values.
%Consider a network $\mathcal{G}$ with $m$ stubborn agents such that each stubborn agent is an LTP agent. Further consider that the agents in $\mathcal{G}$ can be partitioned into $m$ groups ($\mathcal{V}_1,...,\mathcal{V}_m$) with each group containing an LTP agent and the nodes persuaded by it. 
Let the agents be re-numbered such that the agents in subgroup $\mathcal{V}_i$ are indexed as $\sum_{j=1}^{i-1}|\mathcal{V}_j|+1,...,\sum_{j=1}^{i}|\mathcal{V}_j|$ for $i\in[m]$. %are agents , $|\mathcal{V}_1|+1,...,|\mathcal{V}_1|+|\mathcal{V}_2|$ are agents in subgroup $\mathcal{V}_2$ and so on. 
Additionally, let the stubborn agent in the subgroup $\mathcal{V}_i$ be indexed as $\sum_{j=1}^{i-1}|\mathcal{V}_j|+1$. After a permutation of the rows and columns of $W$, it can be written as: $\scriptsize W=\begin{bmatrix}
    W_{11} & ... & W_{1m} \\
    \vdots & \ddots & \vdots \\
     W_{m1} & ... & W_{mm}   
     \end{bmatrix}$.

Similarly, $\mathbf{x}^*$ is partitioned as $\mathbf{x}^*=[\mathbf{x}_1^*,\mathbf{x}_2^*,..., \mathbf{x}_m^*]$ with $\mathbf{x}_i^*$ denoting the final opinions of subgroup $\mathcal{V}_i$. Since each subgroup forms an opinion cluster, $\mathbf{x}_i^*=y_i\mathbb{1}_{|\mathcal{V}_i|}$ where $y_i\in \mathbb{R}$. For the given case, the final opinions of agents in $\mathcal{G}$ satisfy:
\begingroup
\setlength{\abovedisplayskip}{2pt}
\setlength{\belowdisplayskip}{2pt}
\begin{align}
\label{eqn:steady_state_OC}
    {y}_1 \mathbb{1}_{|\mathcal{V}_1|}&=(I-\tilde{\beta}_1)(W_{11}y_1\mathbb{1}_{|\mathcal{V}_1|}+...+W_{1m}{y}_m \mathbb{1}_{|\mathcal{V}_m|})+\tilde{\beta}_1\mathbf{x}_1(0) \nonumber\\
      {y}_2 \mathbb{1}_{|\mathcal{V}_2|}&=(I-\tilde{\beta}_2)(W_{21}y_1\mathbb{1}_{|\mathcal{V}_1|}+...+W_{2m}{y}_m \mathbb{1}_{|\mathcal{V}_m|})+\tilde{\beta}_2\mathbf{x}_2(0) \nonumber\\
        &\vdots \nonumber \\
         {y}_m \mathbb{1}_{|\mathcal{V}_m|}&=(I-\tilde{\beta}_m)(W_{m1}y_1\mathbb{1}_{|\mathcal{V}_1|}+...+W_{mm}{y}_m \mathbb{1}_{|\mathcal{V}_m|})+ \nonumber\\
         &\tilde{\beta}_m\mathbf{x}_m(0)
\end{align}
\endgroup

where $\tilde{\beta_i}=\operatorname{diag}([a_{i},\mathbb{0}_{|\mathcal{V}_i-1|}])$, $a_i=\beta_{\sum_{j=1}^{i-1}|\mathcal{V}_j|+1}$ and $\mathbf{x}_i(0)$ denotes initial opinions of agents in subgroup $\mathcal{V}_i$ for $i\in[m]$. 
Under the given indexing, the first row of $W_{ij}$ (where $i\neq j$) corresponds to stubborn LTP agent $s_i$, and the remaining rows correspond to the agents in $\mathcal{N}_{s_i}$. Since the agents in $\mathcal{N}_{s_i}$ cannot have an incoming edge from an agent in any other subgroup $\mathcal{V}_j$, $W_{ij}\mathbb{1}_{|\mathcal{V}_j|}=\begin{bmatrix}
        \omega_{ij} & \mathbb{0}_{|\mathcal{V}_1|-1}
    \end{bmatrix}^T$. %Thus, eqn. \eqref{eqn:steady_state_OC}, can be further simplified to,
   % we use the following property that each $W_{ij}$ for $i\neq j$ satisfies: . Based on the path-based definition of LTPs,  
   Thus, from eqn. \eqref{eqn:steady_state_OC}, we obtain :

$  \scriptsize  \begin{bmatrix}
        y_1 \\ y_2 \\ \vdots \\ y_m
    \end{bmatrix}=(I-A) \begin{bmatrix}
        \omega_{11} & \omega_{12} & ... & \omega_{1m} \\
        \omega_{21} & \omega_{22} & ... & \omega_{2m} \\
        \vdots & \cdots & \ddots & \vdots \\
        \omega_{m1} & \omega_{m2} & ... & \omega_{mm} \\ 
    \end{bmatrix}\begin{bmatrix}
        y_1 \\ y_2 \\ \vdots \\ y_m
    \end{bmatrix}+A\mathbf{x}_s(0)
$ 

where $\omega_{ii}$ is the row sum of the first row of $W_{ii}$ for $i\in[m]$ and $A=\operatorname{diag}(a_1,...,a_m)$. Let $\Omega=[\omega_{ij}]$. Clearly, $\Omega \mathbb{1}_m=\mathbb{1}_m$. Moreover, $A$ satisfies $ \mathbb{0}_m <A\mathbb{1} \leq \mathbb{1}_m$. Thus, $\rho((I-A)\Omega)<1$. We can express $\mathbf{x}_s(0)$ as $\mathbf{x}_s(0)=A^{-1}(I-(I-A)\Omega)A\mathbf{y}$, where $\mathbf{y}=[y_1,...,y_m]$. Next, we prove the following claim:

\textbf{Claim:} $y_i\neq y_j$ for almost all $\mathbf{x}_s(0)$, where $i,j\in[m]$ and $i\neq j$.

WLOG, let $i=1$ and $j=2$. Consider a set of points in $\mathbb{R}^m$ where $y_{1}=y_{2}$. This set forms an $m-1$ dimensional subspace of $\mathbb{R}^m$ with basis $[1,1,0,...,0],\mathbb{e}_3,...,\mathbb{e}_m$. %Thus, the subspace has a dimension $m-1$. 
Consequently, this set is a measure zero set. 
We can form $\binom{m}{2}$ such subspaces where $y_i=y_j$ for any $i,j\in[m]$. Since the union of finitely many measure-zero sets has measure zero, the set of all such points $\mathbf{y}\in \mathbb{R}^m$ for which the final opinions of agents in two or more distinct subgroups is equal forms a measure zero set. 

For any linear transformation $L:\mathbb{R}^m \to \mathbb{R}^m$ and a measurable set $E$, the following holds: $m(L(E))=\operatorname{det}(L) \cdot m(E)$, where $m(.)$ denotes the measure of a set (Ch.2 in \cite{stein2009real}). Let $E$ be the set of points such that there is a $i,j\in[m]$ such that $y_i=y_j$. Then, $m(E)$ is zero. Since
$A^{-1}(I-(I-A)\Omega)$ is an invertible matrix with a non-zero determinant, $m(A^{-1}(I-(I-A)\Omega)E)=0$.
Therefore, the initial states for which any two opinion clusters converge to the same opinion form a measure-zero set. Consequently, cluster consensus is achieved for almost all initial opinions.

\textbf{(Necessity)}
Let a network $\mathcal{G}$ with $m$ stubborn agents achieve cluster consensus with $m$ opinion clusters for almost all initial opinions. By Theorem \ref{thm:WC}, it follows that the agents in $\mathcal{G}$ must form $m$ disjoint groups such that each subgroup is composed of the stubborn agent, which is the LTP, and the rest are persuaded by it. \hfill$\blacksquare$

%each of the $m$ opeach subthe $m$ disjoint subgroups, formed by the agents in $\mathcal{G}$, such that These $m$ subgroups form the $m$ opinion clusters. Further, we know that each such opinion cluster converges to a distinct opinion for almost all initial states.

%each %stubborn agent belongs to a distinct subgroup of agents that forms an opinion cluster. Hence, each opinion cluster are

%If $\mathcal{G}$ has oblivious agents and it has multiple iSCCs, then in most scenarios, the agents in each iSCC converge have consensus and converge to distinct values. However, if two such iSCCs exist that have identical number of nodes and equal edge weights, then these agents can form in both the iSCCs converge to the same opinion value. Thus, only when the set of oblivious agents contain a unique aperiodic iSCC it converges to a consensus value that always differs from the rest of the opinion clusters.

%and $x_p^*=-(\sum_{i=1}^{m}r_{p,n+i}^1x_{n+i}(0))/(r_{p,p}^1+r_{p,q}^1)$.
% \begin{proof}
\vspace{-10pt}
\bibliographystyle{IEEEtran}
\bibliography{IEEEabrv,references-2}

\begin{thebibliography}{10}
\providecommand{\url}[1]{#1}
\csname url@samestyle\endcsname
\providecommand{\newblock}{\relax}
\providecommand{\bibinfo}[2]{#2}
\providecommand{\BIBentrySTDinterwordspacing}{\spaceskip=0pt\relax}
\providecommand{\BIBentryALTinterwordstretchfactor}{4}
\providecommand{\BIBentryALTinterwordspacing}{\spaceskip=\fontdimen2\font plus
\BIBentryALTinterwordstretchfactor\fontdimen3\font minus \fontdimen4\font\relax}
\providecommand{\BIBforeignlanguage}[2]{{%
\expandafter\ifx\csname l@#1\endcsname\relax
\typeout{** WARNING: IEEEtran.bst: No hyphenation pattern has been}%
\typeout{** loaded for the language `#1'. Using the pattern for}%
\typeout{** the default language instead.}%
\else
\language=\csname l@#1\endcsname
\fi
#2}}
\providecommand{\BIBdecl}{\relax}
\BIBdecl

\bibitem{pratelli2024entropy}
M.~Pratelli, F.~Saracco, and M.~Petrocchi, ``Entropy-based detection of twitter echo chambers,'' \emph{PNAS nexus}, vol.~3, no.~5, p. pgae177, 2024.

\bibitem{echo_chamber}
Y.~Ge, S.~Zhao, H.~Zhou, C.~Pei, F.~Sun, W.~Ou, and Y.~Zhang, ``Understanding echo chambers in e-commerce recommender systems,'' in \emph{Proceedings of the 43rd International ACM SIGIR Conference on Research and Development in Information Retrieval}, ser. SIGIR '20, 2020, p. 2261–2270.

\bibitem{xia2011clustering}
W.~Xia and M.~Cao, ``Clustering in diffusively coupled networks,'' \emph{Automatica}, vol.~47, no.~11, pp. 2395--2405, 2011.

\bibitem{gambuzza2020distributed}
L.~V. Gambuzza and M.~Frasca, ``Distributed control of multiconsensus,'' \emph{IEEE Transactions on Automatic Control}, vol.~66, no.~5, pp. 2032--2044, 2020.

\bibitem{Switching_T}
A.~Bizyaeva, G.~Amorim, M.~Santos, A.~Franci, and N.~E. Leonard, ``Switching transformations for decentralized control of opinion patterns in signed networks: Application to dynamic task allocation,'' \emph{IEEE Control Systems Letters}, vol.~6, pp. 3463--3468, 2022.

\bibitem{anderson2008rigid}
B.~D. Anderson, C.~Yu, B.~Fidan, and J.~M. Hendrickx, ``Rigid graph control architectures for autonomous formations,'' \emph{IEEE Control systems magazine}, vol.~28, no.~6, pp. 48--63, 2008.

\bibitem{tomaselli2023multiconsensus}
C.~Tomaselli, L.~V. Gambuzza, F.~Sorrentino, and M.~Frasca, ``Multiconsensus induced by network symmetries,'' \emph{Systems \& Control Letters}, vol. 181, p. 105629, 2023.

\bibitem{degroot1974consensus}
M.~H. DeGroot, ``Reaching a consensus,'' \emph{Journal of the American Statistical Association}, vol.~69, no. 345, pp. 118--121, 1974.

\bibitem{friedkin1990opinions}
N.~E. Friedkin and E.~C. Johnsen, ``Social influence and opinions,'' \emph{The Journal of Mathematical Sociology}, vol.~15, pp. 193--206, 1990.

\bibitem{rainer2002opinion}
H.~Rainer and U.~Krause, ``Opinion dynamics and bounded confidence: Models, analysis and simulation,'' \emph{Journal of Artificial Societies and Social Simulation}, vol.~5, no.~3, 2002.

\bibitem{dandekar2013biased}
P.~Dandekar, A.~Goel, and D.~T. Lee, ``Biased assimilation, homophily, and the dynamics of polarization,'' \emph{Proceedings of the National Academy of Sciences}, vol. 110, no.~15, pp. 5791--5796, 2013.

\bibitem{friedkin2015control}
N.~E. Friedkin, ``The problem of social control and coordination of complex systems in sociology: A look at the community cleavage problem,'' \emph{IEEE Control Systems Magazine}, vol.~35, pp. 40--51, 2015.

\bibitem{parsegov2016novel}
S.~E. Parsegov, A.~V. Proskurnikov, R.~Tempo, and N.~E. Friedkin, ``Novel multidimensional models of opinion dynamics in social networks,'' \emph{IEEE Transactions on Automatic Control}, vol.~62, no.~5, pp. 2270--2285, 2016.

\bibitem{TIAN2018213}
Y.~Tian and L.~Wang, ``Opinion dynamics in social networks with stubborn agents: An issue-based perspective,'' \emph{Automatica}, vol.~96, pp. 213--223, 2018.

\bibitem{como2016local}
G.~Como and F.~Fagnani, ``From local averaging to emergent global behaviors: The fundamental role of network interconnections,'' \emph{Systems \& Control Letters}, vol.~95, pp. 70--76, 2016.

\bibitem{Opinion_Fluctuations}
D.~Acemo\u{g}lu, G.~Como, F.~Fagnani, and A.~Ozdaglar, ``Opinion fluctuations and disagreement in social networks,'' \emph{Mathematics of Operations Research}, vol.~38, no.~1, pp. 1--27, 2013.

\bibitem{yao2022cluster}
L.~Yao, D.~Xie, and J.~Zhang, ``Cluster consensus of opinion dynamics with stubborn individuals,'' \emph{Systems \& Control Letters}, vol. 165, p. 105267, 2022.

\bibitem{GHADERI20143209}
J.~Ghaderi and R.~Srikant, ``Opinion dynamics in social networks with stubborn agents: Equilibrium and convergence rate,'' \emph{Automatica}, vol.~50, no.~12, pp. 3209--3215, 2014.

\bibitem{gionis2013opinion}
A.~Gionis, E.~Terzi, and P.~Tsaparas, ``Opinion maximization in social networks,'' in \emph{Proceedings of the 2013 SIAM international conference on data mining}.\hskip 1em plus 0.5em minus 0.4em\relax SIAM, 2013, pp. 387--395.

\bibitem{ECC_aashi}
A.~Shrinate and T.~Tripathy, ``Towards influence centrality: where to not add an edge in the network?'' in \emph{2025 European Control Conference (ECC)}, 2025, pp. 665--672.

\bibitem{shrinate2025opinionclusteringfriedkinjohnsenmodel}
------, ``Opinion clustering under the {F}riedkin-{J}ohnsen model: Agreement in disagreement,'' {S}eptember 2025, arXiv:2509.11045.

\bibitem{mattei2022bow}
M.~Mattei, M.~Pratelli, G.~Caldarelli, M.~Petrocchi, and F.~Saracco, ``Bow-tie structures of twitter discursive communities,'' \emph{Scientific Reports}, vol.~12, no.~1, p. 12944, 2022.

\bibitem{Kron_red_digraphs}
T.~Sugiyama and K.~Sato, ``Kron reduction and effective resistance of directed graphs,'' \emph{SIAM Journal on Matrix Analysis and Applications}, vol.~44, no.~1, pp. 270--292, 2023.

\bibitem{zhang2006schur}
F.~Zhang, \emph{The Schur complement and its applications}.\hskip 1em plus 0.5em minus 0.4em\relax Springer Science \& Business Media, 2006, vol.~4.

\bibitem{ITALIANO201274}
G.~F. Italiano, L.~Laura, and F.~Santaroni, ``Finding strong bridges and strong articulation points in linear time,'' \emph{Theoretical Computer Science}, vol. 447, pp. 74--84, 2012.

\bibitem{bullo}
F.~Bullo, \emph{Lectures on Network Systems}, {1.6}~ed.\hskip 1em plus 0.5em minus 0.4em\relax Kindle Direct Publishing, 2022.

\bibitem{luo2021cluster}
S.~Luo and D.~Ye, ``Cluster consensus control of linear multiagent systems under directed topology with general partition,'' \emph{IEEE Transactions on Automatic Control}, vol.~67, no.~4, pp. 1929--1936, 2021.

\bibitem{QIN20132898}
J.~Qin and C.~Yu, ``Cluster consensus control of generic linear multi-agent systems under directed topology with acyclic partition,'' \emph{Automatica}, vol.~49, no.~9, pp. 2898--2905, 2013.

\bibitem{maria}
G.~De~Pasquale and M.~E. Valcher, ``Consensus for clusters of agents with cooperative and antagonistic relationships,'' \emph{Automatica}, vol. 135, p. 110002, 2022.

\bibitem{horn2012matrix}
R.~A. Horn and C.~R. Johnson, \emph{Matrix analysis}.\hskip 1em plus 0.5em minus 0.4em\relax Cambridge university press, 2012.

\bibitem{stein2009real}
E.~M. Stein and R.~Shakarchi, \emph{Real analysis: measure theory, integration, and Hilbert spaces}.\hskip 1em plus 0.5em minus 0.4em\relax Princeton University Press, 2009.

\end{thebibliography}
 
\end{document}